\documentclass[10pt,twocolumn,twoside]{IEEEtran}

\IEEEoverridecommandlockouts                              	
\usepackage{amssymb,amsmath,amsfonts,amsthm}
\usepackage{algorithm,algorithmic}
\usepackage{cite}
\usepackage{float}
\usepackage{epsfig}
\usepackage{graphicx}
\usepackage{graphics}
\usepackage{subfigure}
\usepackage{url}
\usepackage{xspace}
\usepackage[usenames,dvipsnames]{color}
\usepackage{soul}
\usepackage{mathtools}
\floatstyle{ruled}
\newfloat{model}{H}{mod}
\floatname{model}{\footnotesize Model}
\newfloat{notatio}{H}{not}
\floatname{notatio}{\footnotesize Notation}

\newenvironment{varalgorithm}[1]
  {\algorithm\renewcommand{\thealgorithm}{#1}}
  {\endalgorithm}

\newenvironment{list3}{
	\begin{list}{$\bullet$}{%
			\setlength{\itemsep}{0.05cm}
			\setlength{\labelsep}{0.2cm}
			\setlength{\labelwidth}{0.3cm}
			\setlength{\parsep}{0in} 
			\setlength{\parskip}{0in}
			\setlength{\topsep}{0in} 
			\setlength{\partopsep}{0in}
			\setlength{\leftmargin}{0.22in}}}
	{\end{list}}

\newenvironment{list4}{
	\begin{list}{$\bullet$}{%
			\setlength{\itemsep}{0.05cm}
			\setlength{\labelsep}{0.2cm}
			\setlength{\labelwidth}{0.3cm}
			\setlength{\parsep}{0in} 
			\setlength{\parskip}{0in}
			\setlength{\topsep}{0in} 
			\setlength{\partopsep}{0in}
			\setlength{\leftmargin}{0.16in}}}
	{\end{list}}

\newenvironment{list4a}{
	\begin{list}{$\bullet$}{%
			\setlength{\itemsep}{0.05cm}
			\setlength{\labelsep}{0.2cm}
			\setlength{\labelwidth}{0.3cm}
			\setlength{\parsep}{0in} 
			\setlength{\parskip}{0in}
			\setlength{\topsep}{0in} 
			\setlength{\partopsep}{0in}
			\setlength{\leftmargin}{0.16in}}}
	{\end{list}}

\newenvironment{list5}{
	\begin{list}{$\bullet$}{%
			\setlength{\itemsep}{0.05cm}
			\setlength{\labelsep}{0.2cm}
			\setlength{\labelwidth}{0.3cm}
			\setlength{\parsep}{0in} 
			\setlength{\parskip}{0in}
			\setlength{\topsep}{0in} 
			\setlength{\partopsep}{0in}
			\setlength{\leftmargin}{0.18in}}}
	{\end{list}}

\usepackage{url,changebar,bm,xspace,dsfont}
\let\mathbb=\mathds % I much prefer the dsfont over the bbfont
\def\day{\mathrm{day}} % for differentials

\newtheorem{theorem}{Theorem}

\newtheorem{defn}{Definition}
\newtheorem{prop}{Proposition}

\newtheorem{remark}{Remark}

\ifCLASSINFOpdf
 \else
\fi

\usepackage{amsmath,amsfonts,amssymb,amscd}
\usepackage{capt-of}
\usepackage{color}
\usepackage{cite}
\usepackage{verbatim}
\usepackage{hyperref}

\hypersetup{
     colorlinks = true,
     linkcolor = [rgb]{0,0,1},
     anchorcolor = [rgb]{0,0,1},
     citecolor = [rgb]{0.9,0.5,0},
     filecolor = [rgb]{0,0,1},
     pagecolor = [rgb]{1,1,0},
     urlcolor = [rgb]{0.9,0.5,0},
     bookmarks,
        bookmarksopen = true,
        bookmarksnumbered = true,
        breaklinks = true,
        linktocpage,
        pagebackref,
        colorlinks = true,
        linkcolor = [rgb]{0.2,0.6,0.2},
        urlcolor  = [rgb]{0,0,1},
        citecolor = [rgb]{0.9,0.5,0},
        anchorcolor = [rgb]{0.2,0.6,0.2},
        hyperindex = true,
        hyperfigures
}
\newtheorem{lemma}{\bfseries Lemma}

\begin{document}

\title{\LARGE \bf Distributed Event-Triggered  Algorithms for \\ Finite-Time Privacy-Preserving Quantized Average Consensus}
%\author{ \parbox{3 in}{\centering Huibert Kwakernaak*
%         \thanks{*Use the $\backslash$thanks command to put 
%information here}\\
%         Faculty of Electrical Engineering, Mathematics and 
%Computer Science\\
%         University of Twente\\
%         7500 AE Enschede, The Netherlands\\
%         {\tt\small h.kwakernaak@autsubmit.com}}
%         \hspace*{ 0.5 in}
%         \parbox{3 in}{ \centering Pradeep Misra**
%         \thanks{**The footnote marks may be inserted manually}\\
%        Department of Electrical Engineering \\
%         Wright State University\\
%         Dayton, OH 45435, USA\\
%         {\tt\small pmisra@cs.wright.edu}}
%}

\author{Apostolos~I.~Rikos, Themistoklis Charalambous, Karl~H.~Johansson, and Christoforos~N.~Hadjicostis
\thanks{Apostolos~I.~Rikos and K.~H.~Johansson are with the Division of Decision and Control Systems, KTH Royal Institute of Technology, SE-100 44 Stockholm, Sweden. E-mails: {\tt \{rikos,kallej\}@kth.se}.}
\thanks{T. Charalambous is with the Department of Electrical Engineering and Automation, Aalto University, 02150 Espoo, Finland.  E-mail:{\tt~themistoklis.charalambous@aalto.fi}.}
\thanks{C.~N.~Hadjicostis is with the Department of Electrical and Computer Engineering, University of Cyprus, 1678 Nicosia, Cyprus: E-mail:{\tt~chadjic@ucy.ac.cy}.}
\thanks{Preliminary results of this work were presented at the 2020 IEEE Conference on Decision and Control \cite{2020:Rikos_Privacy_CDC}. 
We extend these results by (i) proposing a second privacy-preserving algorithm and deriving its required topological conditions, (ii) extending the simulations, where we also compare the two proposed algorithms, and (iii) presenting an application for computing the power requests in smart grids in a privacy-preserving manner.}
\thanks{This work was supported in part by the European Union's Horizon 2020 research and innovation program under grant agreement No~739551 (KIOS~CoE).}
}

\maketitle
\thispagestyle{empty}
\pagestyle{empty}

% ===============================================
%
%
% ABSTRACT
%
%
% ===============================================
\begin{abstract}

In this paper, we consider the problem of privacy preservation in the average consensus problem when communication among nodes is quantized. 
More specifically, we consider a setting where some nodes in the network are curious but not malicious and they try to identify the initial states of other nodes based on the data they receive during their operation (without interfering in the computation in any other way), while some nodes in the network want to ensure that their initial states cannot be inferred exactly by the curious nodes.  
We propose two privacy-preserving event-triggered quantized average consensus algorithms that can be followed by any node wishing to maintain its privacy and not reveal the initial state it contributes to the average computation. 
Every node in the network (including the curious nodes) is allowed to execute a privacy-preserving algorithm or its underlying average consensus algorithm. 
In the first algorithm, each node initially injects a quantized offset and continues injecting offsets every time a certain event-triggered condition is satisfied, such that, after a finite number of events, the accumulated injected offset becomes equal to zero. 
In the second algorithm, each node injects a quantized offset to its out-neighboring nodes at the algorithm's initialization, such that the accumulated sum is equal to zero. 
Under certain topological conditions, both algorithms allow the nodes who adopt privacy-preserving protocols to preserve the privacy of their initial quantized states and at the same time to obtain, after a finite number of steps, the exact average of the initial states while processing and transmitting \textit{quantized} information. 
Illustrative examples demonstrate the validity and performance of our proposed algorithms. 
A motivating application is presented in which smart meters in a smart grid collect real-time demands for power in a neighborhood and the aggregated demand is distributively computed in a privacy-preserving manner.
\end{abstract}


% ===============================================
%
%
% INTRODUCTION
%
%
% ===============================================
\section{Introduction}\label{intro}

%\todo{REPLACE PANDEMIC WITH POWER NETWORKS}

A problem of particular interest in distributed control is the \textit{consensus} problem, in which nodes communicate locally with other nodes under constraints on connectivity \cite{2004:Murray}. 
In distributed averaging (a special case of the consensus problem), each node that is initially endowed with a numerical state, which it updates in an iterative fashion by sending/receiving information to/from other neighboring nodes, eventually computes the average of all initial states. 
Average consensus %has received significant attention recently and 
has been studied extensively in settings where each node processes and transmits real-valued states with infinite precision; see, for example, \cite{2018:BOOK_Hadj} and references therein.

The case where capacity-limited network links can only allow messages of certain length to be transmitted between nodes has also received significant attention recently, as it effectively extends techniques for average consensus towards quantized consensus. 
Quantized processing and communication is better suited to the available network resources (e.g., physical memories of finite capacity and digital communication channels of limited data rate), while it also exhibits other advantages such as amenability to security and privacy enhancements \cite{2019:Alexandru}. 
For example, public-key cryptosystems require integer numbers to operate with, since non-quantized consensus algorithms would be subject to quantization errors in the final result \cite{1999:Paillier, Ruan:2019, hadjicostis2020privacy}. 
%For these reasons, several works studied the quantized average consensus problem and various probabilistic and deterministic strategies have been proposed 
For these reasons, several probabilistic and deterministic strategies have been proposed for solving the quantized average consensus problem
\cite{2007:Aysal_Rabbat, 2012:Lavaei_Murray , 2011:Cai_Ishii, 2016:Chamie_Basar, 2020:Rikos_Mass_Accum}.

Average consensus algorithms require each node to exchange and disclose state information to its neighbors. 
This may be undesirable in case the state of some nodes is private or contains sensitive information. 
Additionally, in many occasions there might be nodes in the network that are curious and aim to extract private and/or sensitive information.
In many emerging applications (e.g., health care and opinion forming over social networks) preserving the privacy of participating components is necessary for enabling cooperation between nodes without requiring them to disclose sensitive information. 
There have been different approaches for dealing with privacy preservation in such systems. 
For example, \cite{Kefayati:2007} proposed a method in which each node wishing to protect its privacy adds a random offset value to its initial state, thus ensuring that its true state will not be revealed to curious nodes that might be observing the exchange of data in the network. 
The main idea is based upon the observation that when a large number of nodes employ the protocol, the sum of their offsets will be essentially zero and therefore the nodes will converge to the true average state of the network. 
%\cite{2008:Dwork, 2016:CortesPappas}
A related line of research is based on differential privacy \cite{2016:CortesPappas, NOZARI:2017}, in which nodes inject uncorrelated noise into the exchanged messages so that the data associated to a particular node cannot be inferred by a curious node during the execution of the algorithm. 
However, the exact average state is not eventually obtained due to the induced trade-off between privacy and computational accuracy \cite{NOZARI:2017}. 
To overcome this trade-off and guarantee convergence to the exact average, the injection of correlated noise at each time step and for a finite period of time was proposed in \cite{2013:Nikolas_Hadj}, thus allowing a node to avoid revealing its own initial state or the initial states of other nodes. 
Once this period of time ends, each node ensures that the accumulated sum of offsets it added in the iterative computation is removed. 
In \cite{Mo-Murray:2017}, the nodes asymptotically subtract the initial offset values they added in the computation while in \cite{2017:Gupta} each node masks its initial state with an offset such that the sum of the offsets of each node is zero, thus guaranteeing convergence to the average. 
Another approach that guarantees privacy preservation is via homomorphic encryption \cite{ChrisCDC:2018, Ruan:2019, hadjicostis2020privacy}. 
However, this approach requires the existence of trusted nodes and imposes heavier computational requirements on the nodes.

The main contributions of this paper are the following:
\begin{list4}
\item[i)] two novel distributed algorithms which achieve quantized average consensus under privacy constraints that converge after a finite number of time steps, and 
\item[ii)] their application to power request in smart grids under privacy-preserving guarantees.
\end{list4}

\subsection{Distributed quantized average consensus algorithms}

%Specifically, each node that executes the privacy preserving algorithms is able to achieve quantized average consensus and ensure that its initial state cannot be inferred exactly by curious nodes which receive and transmit data while they execute a privacy-preserving algorithm or its underlying average consensus algorithm. 
During its operation, each node that would like to protect its privacy from other curious (but not malicious) nodes follows one of the two finite-time event-triggered quantized average consensus protocols. 
The privacy preserving algorithms essentially involve adding and subtracting offsets to each node's state in two different ways: 
\\ \noindent {\em 1. The first algorithm injects offsets according to an event-based strategy for a predefined number of steps.}
Specifically, when the token that triggers action arrives at a specific node for the first time, the node adds a substantial negative \emph{quantized} offset to its initial state. 
This initial offset is determined by the node at the first triggering, and is gradually removed at later triggerings (when certain conditions are satisfied) ensuring that the total accumulated sum of injected offsets is canceled out.
\\ \noindent {\em 2. The second algorithm injects offsets only during the initialization procedure.}
Specifically, each node injects a quantized offset to the states of its out-neighboring nodes only during the initialization procedure. 
Then, the node injects to its own state a (possibly) nonzero offset (in addition to any offsets injected by its in-neighbors) such that the accumulated sum of the injected offsets is equal to zero, and proceeds with executing a finite-time event-triggered quantized average consensus protocol.

%We present and analyze the operation of both algorithms (the first algorithm in Section~\ref{sec:distr_algo} and the second algorithm in Section~\ref{sec:distr_algo2}).  More specifically, show that both algorithms converge after a finite number of time steps (see Lemma~\ref{first_lemma}, Lemma~\ref{second_lemma}, Theorem~\ref{Privacy_Conver_Quant_Av} for the first algorithm and Theorem~\ref{Privacy_Conver_Quant_Av_Alg2} for the second algorithm); we also present the topological conditions that ensure privacy for the nodes that follow the proposed protocols (in Proposition~\ref{prop:1} for the first algorithm and in Proposition~\ref{prop_privacy_2} for the second algorithm). Then, we present numerical simulations in which we demonstrate and compare the operation of the two algorithms over random digraphs (in Section~\ref{sec:results}). Also, we present an application on power requests in smart grids (in Section~\ref{app}). 

We show that both algorithms converge after a finite number of time steps. We also present the topological conditions that ensure privacy for the nodes that follow the proposed protocols. Then, we present numerical simulations in which we demonstrate and compare the operation of the two algorithms over random digraphs.

Note here that the algorithms presented in this paper build on the algorithms introduced in \cite{2013:Nikolas_Hadj, Mo-Murray:2017}. 
However, unlike other privacy preserving protocols proposed in the literature (e.g., \cite{Kefayati:2007, 2013:Nikolas_Hadj, Mo-Murray:2017, 2019:Allerton_themis}), the proposed algorithms take advantage of their finite time operation since the added offsets are integers. %not real numbers. 
As a result, consensus to the \textit{exact} average of the initial states is achieved after a finite number of steps, while the error, introduced from the offset, vanishes completely. 
%Furthermore, the first algorithm of this paper is presented in \cite{2020:Rikos_Privacy_CDC}. 
%In this paper we extend \cite{2020:Rikos_Privacy_CDC} by proposing a second privacy-preserving algorithm, by deriving its required topological conditions, by presenting extended simulations (where we also compare the two proposed algorithms) and by presenting an application for calculating the number of people infected from a contagious disease in a privacy-preserving manner. 

\subsection{Application: Power Request in Smart Grids under Privacy-Preserving Guarantees}
\label{app}

%In this paper we also present a motivating application of our proposed algorithms over a network of smart grids. 
Smart grids are considered as the next-generation power supply networks \cite{2014:Chao}. 
In this application, a neighborhood of interconnected households is able to request the \emph{total} demanded power from a smart meter in a privacy-preserving manner. 
One of the main characteristics of smart grids is that the power generator produces electricity based on consumers requests which are generated in real-time and collected by smart meters. 
%This means that the consumers requests are satisfied while excess electricity generation is avoided. 
%The first is an important characteristic which guarantees proper function of the power grid while the second is significant for environmental purposes and profit maximization. 
Real time power demand data may contain patterns from daily/weekly life schedule. %which should be kept private for several reasons. 
Since potential leakage of this sensitive information may lead to malicious situations against the residents of specific households (e.g., the probability that thieves will attempt breaking into the household may increase), it is essential to preserve the privacy of the data sent to smart meters (which contains each household's daily requested power). 

During the operation of smart grids we have the following sequence of actions: i) the smart meter collects the daily demands from each household (and possibly electric cars in charging lots) and transmits them to the power generator; ii) having received the demands the power generator produces and delivers the demanded electricity to each substation in the corresponding region; iii) the electricity is claimed from the substation directly, which delivers the demanded electricity to the households without any other entities having access to this transaction. 
The charging of the demanded electricity can be communicated at the end of the month from the substations.

As an example, let us consider in Fig.~\ref{app_fig} a neighborhood with $8$ households denoted as $\mathcal{B} = \{ b_1, b_2, ..., b_8 \}$, a smart meter denoted as $v_{\mathrm{sm}}$, a substation denoted as $v_{\mathrm{Sub}}$ and a power generator denoted as $v_{\mathrm{PGen}}$. During the operation of the proposed algorithms the smart meter $v_{\mathrm{sm}}$ collects (through say $b_1$) 
%the distorted daily power demands $\widetilde{\mathcal{C}}^{\day}$ of every household in order to provide the billing at the end of the month for every household.  Furthermore, it also collects 
the state variable of household $b_1$ which is equal to the average of the daily demanded power from each household in the neighborhood.
Then, it multiplies it with the number of houses in the neighborhood in order to calculate the total demanded power and it transmits the total demanded power to the power generator. 

\begin{figure}[!h]
\begin{center}
\includegraphics[width=0.775\columnwidth]{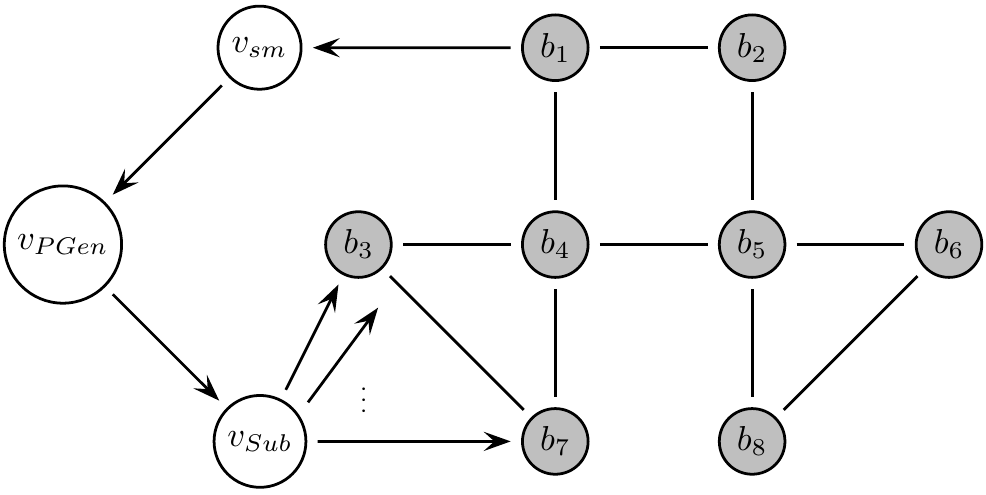}
\caption{Example of a digraph representing a smart grid consisting of a neighborhood with $8$ households $b_1 - b_8$, a smart meter $v_{\mathrm{sm}}$, a substation $v_{Sub}$, and a power generator $v_{\mathrm{PGen}}$.}
\label{app_fig}
\end{center}
\end{figure}

For a description of how our algorithms have been applied in this context, please refer to Appendix \ref{app_analysis}. 

\subsection{Organization of the paper}

The rest of the paper is organized as follows. 
In Section~\ref{app}, we  present a motivating application on power requests in smart grids. 
In Section~\ref{sec:preliminaries}, we review necessary notation and background, while in Section~\ref{sec:probForm} we provide the problem formulation. 
In Section \ref{sec:distr_algo} we present our first privacy strategy along with the corresponding distributed algorithm. 
Furthermore, we analyze the convergence of our algorithm while we present sufficient topological conditions that ensure privacy preservation. 
In Section \ref{sec:distr_algo2} we present our second privacy strategy and the corresponding distributed algorithm while, we analyze its convergence and we present sufficient topological conditions that ensure privacy preservation. 
In Section~\ref{sec:results} we demonstrate our strategies via illustrative examples.
%for calculating, under privacy-preserving guarantees, the average number of people infected from an infectious disease (that also imposes constraints on movement and interconnectivity, such as COVID-$19$). 
In Section~\ref{sec:conclusions} we draw concluding remarks and discuss future directions. 
Finally, the proofs of lemmas and theorems are provided in appendices.

% ===============================================
%
%
% NOTATION
%
%
% ===============================================
\section{Notation and Preliminaries}\label{sec:preliminaries}

\subsection{Notation}

The sets of real, rational, integer, and natural numbers are denoted by $ \mathbb{R}, \mathbb{Q}, \mathbb{Z}$, and $\mathbb{N}$, respectively. 
The symbols $\mathbb{Z}_{\geq 0}$ ($\mathbb{Z}_{>0}$) and $\mathbb{Z}_{\leq 0}$ ($\mathbb{Z}_{<0}$) denote the sets of nonnegative (positive) and nonpositive (negative) integers respectively. 
Vectors are denoted by small letters whereas matrices are denoted by capital letters. 
The transpose of a matrix $A$ is denoted by $A^T$. For $A\in \mathbb{R}^{n\times n}$, $A_{ij}$ denotes the entry at row $i$ and column $j$. By $\mathbf{1}$ we denote the all-ones vector and by $I$ we denote the identity matrix (of appropriate dimensions).

\subsection{Graph Theory}

Consider a network of $n$ ($n \geq 2$) nodes communicating only with their immediate neighbors. 
The communication topology can be captured by a directed graph (digraph), called \textit{communication digraph}. 
A digraph is defined as $\mathcal{G}_d = (\mathcal{V}, \mathcal{E})$, where $\mathcal{V} =  \{v_1, v_2, \dots, v_n\}$ with cardinality $n  = | \mathcal{V} | \geq 2 $ is the set of nodes and $\mathcal{E} \subseteq \mathcal{V} \times \mathcal{V} - \{ (v_j, v_j) \ | \ v_j \in \mathcal{V} \}$ is the set of edges (self-edges excluded) whose cardinality is denoted as $m = | \mathcal{E} |$. 
A directed edge from node $v_i$ to node $v_j$ is denoted by $m_{ji} \triangleq (v_j, v_i) \in \mathcal{E}$, and captures the fact that node $v_j$ can receive information from node $v_i$ (but not the other way around). 
We assume that the given digraph $\mathcal{G}_d = (\mathcal{V}, \mathcal{E})$ is \textit{strongly connected} (i.e., for each pair of nodes $v_j, v_i \in \mathcal{V}$, $v_j \neq v_i$, there exists a directed \textit{path}\footnote{A directed \textit{path} from $v_i$ to $v_j$ exists if we can find a sequence of nodes $v_i \equiv v_{l_0},v_{l_1}, \dots, v_{l_t} \equiv v_j$ such that $(v_{l_{\tau+1}},v_{l_{\tau}}) \in \mathcal{E}$ for $ \tau = 0, 1, \dots , t-1$.} from $v_i$ to $v_j$). 
The subset of nodes that can directly transmit information to node $v_j$ is called the set of in-neighbors of $v_j$ and is represented by $\mathcal{N}_j^- = \{ v_i \in \mathcal{V} \; | \; (v_j,v_i)\in \mathcal{E}\}$, while the subset of nodes that can directly receive information from node $v_j$ is called the set of out-neighbors of $v_j$ and is represented by $\mathcal{N}_j^+ = \{ v_l \in \mathcal{V} \; | \; (v_l,v_j)\in \mathcal{E}\}$. 
The cardinality of $\mathcal{N}_j^-$ is called the \textit{in-degree} of $v_j$ and is denoted by $\mathcal{D}_j^-$ (i.e., $\mathcal{D}_j^- = | \mathcal{N}_j^- |$), while the cardinality of $\mathcal{N}_j^+$ is called the \textit{out-degree} of $v_j$ and is denoted by $\mathcal{D}_j^+$ (i.e., $\mathcal{D}_j^+ = | \mathcal{N}_j^+ |$).

\subsection{Node Operation}\label{subsec:not_priv_alg}

With respect to quantization of information flow, we have that at time step $k \in \mathbb{Z}_{\geq 0}$ each node $v_j \in \mathcal{V}$ maintains the state variables $y^s_j[k], z^s_j[k], q_j^s[k]$, where $y^s_j[k] \in \mathbb{Z}$, $z^s_j[k] \in \mathbb{N}$ and  $q_j^s[k] ={y_j^s[k]}/{z_j^s[k]}$, and the mass variables $y_j[k], z_j[k]$, where $y_j[k] \in \mathbb{Z}$ and $z_j[k] \in \mathbb{Z}_{\geq 0}$.
%The aggregate states are denoted by 
%\begin{align*}
%y^s[k] &= [y^s_1[k] \ ... \ y^s_n[k]]^{\rm T} \in \mathbb{Z}^n, \\
%z^s[k] &= [z^s_1[k] \ ... \ z^s_n[k]]^{\rm T} \in \mathbb{N}^n, \\
%q^s[k] &= [q^s_1[k] \ ... \ q^s_n[k]]^{\rm T} \in \mathbb{Q}^n, \\
%y[k] &= [y_1[k] \ ... \ y_n[k]]^{\rm T} \in \mathbb{Z}^n, \\ 
%z[k] &= [z_1[k] \ ... \ z_n[k]]^{\rm T} \in \mathbb{Z}_{\geq 0}^n.
%\end{align*} 
%
%
%\noindent
%\textbf{Transmission Policy.} 
We assume that each node is aware of its out-neighbors and can directly transmit messages to each of them. However, it cannot necessarily receive messages (at least not directly) from them. 
In the proposed distributed protocols, each node $v_j$ assigns a \textit{unique order} in the set $\{0,1,..., \mathcal{D}_j^+ -1\}$ to each of its outgoing edges $m_{lj}$, where $v_l \in \mathcal{N}^+_j$. 
More specifically, the order of link $(v_l,v_j)$ for node $v_j$ is denoted by $P_{lj}$ (such that $\{P_{lj} \; | \; v_l \in \mathcal{N}^+_j\} = \{0,1,..., \mathcal{D}_j^+ -1\}$). 
This unique predetermined order is used during the execution of the proposed distributed algorithm as a way of allowing node $v_j$ to transmit messages to its out-neighbors in a \textit{round-robin}\footnote{When executing the protocol, each node $v_j$ transmits to its out-neighbors, one at a time, by following the predetermined order. The next time it transmits to an out-neighbor, it continues from the outgoing edge it stopped the previous time and cycles through the edges in a round-robin fashion, according to their order.} fashion.

\subsection{Quantized Averaging via Deterministic Mass Summation}\label{Prel_Aver}

The objective of quantized average consensus problems is the development of distributed algorithms which allow nodes to process and transmit quantized information, so that they have short communication packages and eventually obtain, after a finite number of steps, a fraction $q^s$ which is equal to the \textit{exact} average of the initial quantized states of the nodes.

Following the recently proposed method in \cite{2020:Rikos_Mass_Accum}, we assume that each node $v_j$ in the network has a quantized\footnote{Following \cite{2007:Basar, 2011:Cai_Ishii} we assume that the state of each node is integer valued. 
This abstraction subsumes a class of quantization effects (e.g., uniform quantization).} initial state $y_j[0] \in \mathbb{Z}$. 
At each time step $k$, each node $v_j \in \mathcal{V}$ maintains its mass variables $y_j[k] \in \mathbb{Z}$ and $z_j[k] \in \mathbb{Z}_{\geq 0}$, and its state variables $y^s_j[k] \in \mathbb{Z}$, $z^s_j[k] \in \mathbb{N}$ and  $q_j^s[k] = {y_j^s[k]}/{z_j^s[k]}$. 
It updates the values of the mass variables as 
\begin{subequations}\label{Prel_Aver_YZ}
\begin{align}
y_j[k+1] = y_j[k] + \sum_{v_i \in \mathcal{N}_j^-} \mathds{1}_{ji}[k] y_i[k] , \label{subeq:1a} \\
z_j[k+1] = z_j[k] + \sum_{v_i \in \mathcal{N}_j^-} \mathds{1}_{ji}[k] z_i[k] , \label{subeq:1b}
\end{align}
\end{subequations}
where  
\begin{align*}
\mathds{1}_{ji}[k] = 
\begin{cases}
1, & \text{if a message is received at $v_j$ from $v_i$ at $k$,} \\[0.1cm]
0, & \text{otherwise.} 
\end{cases}
\end{align*}
If \emph{any} of the following event-triggered conditions: 
\begin{list3}\label{tr_cond}
\item[(C1):] $z_j[k+1] > z^s_j[k]$, 
\item[(C2):] $z_j[k+1] = z^s_j[k]$ and $y_j[k+1] \geq y^s_j[k]$, 
\end{list3}
is satisfied, node $v_j$ updates its state variables as follows: 
\begin{subequations}\label{State_Prel_Aver_YZ}
\begin{align}
z^s_j[k+1] &= z_j[k+1], \\
y^s_j[k+1] &= y_j[k+1], \\
q^s_j[k+1] &= \frac{y^s_j[k+1]}{z^s_j[k+1]} .
\end{align}
\end{subequations}
Then, it transmits its mass variables $y_j[k+1]$, $z_j[k+1]$ to an out-neighbor $v_l \in \mathcal{N}^+_j$ chosen according to the unique order it assigned to its out-neighbors during initialization and sets its mass variables equal to zero (i.e., $y_j[k+1] = 0$ and $z_j[k+1] = 0$).

\begin{defn}\label{Definition_Quant_Av}
The system is able to achieve quantized average consensus if, for every $v_j \in \mathcal{V}$, there exists $k_0$ so that for every $k \geq k_0$ we have 
\begin{equation}\label{alpha_z_y}
y^s_j[k] = \frac{\sum_{l=1}^{n}{y_l[0]}}{\alpha}  \ \ \text{and} \ \ z^s_j[k] = \frac{n}{\alpha} ,
\end{equation}
for some $\alpha \in \mathbb{N}$. This means that 
\begin{equation}\label{alpha_q}
q^s_j[k] = \frac{(\sum_{l=1}^{n}{y_l[0]}) / \alpha}{n / \alpha} = \frac{\sum_{l=1}^{n}{y_l[0]}}{n} =:  \overline{y} ,
\end{equation}
i.e., for $k \geq k_0$ every node $v_j$ has calculated $\overline{y}$ as the ratio of two integer values. 
\end{defn}

The following result from \cite{2020:Rikos_Mass_Accum} provides a worst case upper bound regarding the number of time steps required for quantized averaging to be achieved.

\begin{theorem}[\hspace{-0.00001cm}\cite{2020:Rikos_Mass_Accum}]
\label{Conver_Quant_Av}
The iterations in (\ref{Prel_Aver_YZ}) and (\ref{State_Prel_Aver_YZ}) allow the set of nodes to reach quantized average consensus (i.e., state variables of each node $v_j \in \mathcal{V}$ fulfil (\ref{alpha_z_y}) and (\ref{alpha_q})) after a finite number of steps $\mathcal{S}_t$, bounded by $\mathcal{S}_t \leq nm^2$, where $n$ is the number of nodes and $m$ is the number of edges in the network. 
\end{theorem}

% ===============================================
%
%
% PROBLEM
%
%
% ===============================================
\section{Problem Formulation}\label{sec:probForm}

Consider a strongly connected digraph $\mathcal{G}_d = (\mathcal{V}, \mathcal{E})$, where each node $v_j \in \mathcal{V}$ has an initial quantized state $y_j[0]$ (for simplicity, we take $y_j[0] \in \mathbb{Z}$). 
Nodes require to calculate 
\begin{equation}\label{init_average}
\overline{y} = \frac{\sum_{j=1}^n y_j[0]}{n}
\end{equation}
in a distributed way, exclusively through local exchange of information. 
The information exchange takes place only between nodes that are neighbors with respect to $\mathcal{G}_d$, which represents the system communication architecture. 
The node set $\mathcal{V}$ is partitioned into three subsets: 1) a subset of nodes $v_j \in \mathcal{V}_p \subset \mathcal{V}$, that wish to preserve their privacy by not revealing their initial states $y_j[0]$ to other nodes, 2) another other subset of nodes $v_c \in \mathcal{V}_c \subset \mathcal{V}$ that are curious and try to identify the initial states $y[0]$ of all or a subset of nodes in the network, and 3) the rest of the nodes  $v_i \in \mathcal{V}_n \subset \mathcal{V}$ that neither wish to preserve their privacy nor identify the states of any other nodes.
%Overall, we have that $\mathcal{V} = P \cup \mathcal{V}_c \cup \mathcal{V}_{np\_nc}$, where $P$ is the set of nodes that wish to preserve the privacy of their initial states, $\mathcal{V}_c$ is the set of curious nodes that try to identify the initial states $y[0]$ of all or a subset of nodes in the network, and $\mathcal{V}_{np\_nc}$ is the set of nodes that are not curious and they do not wish to preserve the privacy of their initial values. 
An example of such a partition is shown in Fig.~\ref{prob_form_graph}.

\begin{figure}[h]
\begin{center}
\includegraphics[width=0.5\columnwidth]{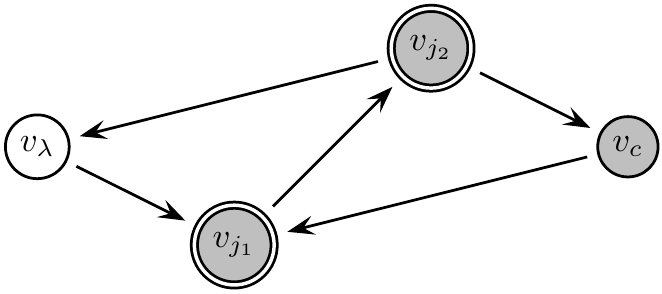}
\caption{Example of a digraph with the different types of nodes in the network: nodes $v_{j_1}, v_{j_2} \in \mathcal{V}_p$ that wish to preserve their privacy, node $v_{c} \in \mathcal{V}_c$ that is curious and wishes to identify the initial states of other nodes in the network, and node $v_{\lambda} \in \mathcal{V}_{n}$ that is neither curious nor wishes to preserve its privacy.}
\label{prob_form_graph}
\end{center}
\end{figure}

The concept of privacy is typically defined as the ability of an individual node to seclude itself or hide information about itself, and thereby express itself selectively. 
In our case, we consider that the information of interest for each node is its initial state $y_j[0]$. 
The notion of privacy that we adopt aims to ensure that the state $y_j[0]$ cannot be inferred exactly by curious nodes and relates to notions of possible innocence in theoretical computer science \cite{1998:Reiter, 2006:Chatzikokolakis} in the sense that there is some uncertainty about $y_j[0]$.

\begin{defn}\label{Definition_Quant_Privacy}
A node $v_j$ is said to preserve the privacy of its initial state $y_j[0] \in \mathbb{Z}$ if the state cannot be inferred exactly by curious nodes at any point during the operation of the protocol. 
\end{defn}

The problem we consider in this paper is to develop a strategy for nodes $v_j \in \mathcal{V}_p$ that wish to prevent their privacy (i.e., not reveal their initial states $y_j[0]$ to other nodes) when they exchange quantized information with neighboring nodes for calculating $\overline{y}$ in \eqref{init_average}. Note that this strategy should allow nodes to operate seamlessly along with those that use the underlying average consensus algorithm described in Section~\ref{subsec:not_priv_alg} that is not privacy-preserving. 
As aforementioned, curious nodes $v_c \in \mathcal{V}_c$ try to identify the initial states $y[0]$ of other nodes $v_j' \in \mathcal{V}_p \cup \mathcal{V}_{n}$ but {\em do not interfere} in the computation in any other way as they execute either a privacy-preserving strategy or its underlying average consensus algorithm in Section~\ref{subsec:not_priv_alg}. 
We also assume that curious nodes $v_c$ may collaborate arbitrarily and that they know the predefined algorithm followed by nodes that would like to preserve their privacy, and the topology of the network, but not the actual parameters chosen by the nodes $v_j\in \mathcal{V}_p$ that want to preserve their privacy. 
Finally, we assume that nodes $v_i \in \mathcal{V}_{n}$ simply execute the underlying average consensus algorithm in Section~\ref{subsec:not_priv_alg}.

%Our contribution is a variant of the quantized averaging strategy (described in Section~\ref{Prel_Aver}), which nodes that wish to preserve their privacy can follow (the remaining nodes simply follow the original strategy in Section~\ref{Prel_Aver}). 

%\vspace{-.3cm}

\begin{remark}
Definition~\ref{Definition_Quant_Privacy} implies that if one looks at the set of variables that are a priori unknown to the curious nodes (e.g., the initial states of nodes, offsets chosen by other nodes, etc.), then one can find at least two different sets of values for these variables that match the observations that become available to the curious nodes (including the eventual knowledge of the average of the initial states of the nodes), such that the node that wants to preserve its privacy exhibits different initial states in these two sets.
\end{remark}

\vspace{-.5cm}

% ===============================================
%
%
% ALGORITHM
%
%
% ===============================================
\section{Event-Based Offset Privacy-Preserving Strategy}
\label{sec:distr_algo}

%In this section, we present and analyze a distributed iterative strategy that allows nodes (while processing and transmitting \textit{quantized} information via available communication links between nodes) to preserve the privacy of their initial quantized states and to obtain, after a finite number of steps, a fraction $q^s$ which is equal to the exact average $\overline{y}$ in \eqref{init_average}. 

\subsection{Initialization for Quantized Privacy-Preserving Strategy}

The primary objective in our system is to calculate $\overline{y}$ in \eqref{init_average} while preserving the privacy of at least the nodes following the protocol. 
Our strategy is based on the event-triggered deterministic algorithm \eqref{Prel_Aver_YZ}--\eqref{State_Prel_Aver_YZ} with some modifications (since the event-triggered deterministic algorithm \eqref{Prel_Aver_YZ}--\eqref{State_Prel_Aver_YZ} is not privacy- preserving). 
The main difference is that a mechanism is deployed that incorporates an offset to the mass variable of each node $v_j \in \mathcal{V}_p$, effectively preserving the privacy of its initial state $y_j[0]$. 

In previous works (see, for example,\cite{Kefayati:2007, 2013:Nikolas_Hadj, Mo-Murray:2017, 2019:Allerton_themis} and references therein) node $v_j$ sets its initial state to $\widetilde{y}_j[0] = y_j[0] + u_j$, where $u_j \in \mathbb{R}$. 
However, in this case, we require that the initial offset $u_j$ is a negative integer number, i.e., $u_j \in \mathbb{Z}_{<0}$, so that the event-triggered conditions (C1) and (C2) are guaranteed to lead to the calculation of the initial average after a finite number of time steps.  
Furthermore, each node $v_j$ maintains the privacy values $u_j[k] \in \mathbb{Z}_{\geq 0}$, the offset adding steps $L_j \in \mathbb{N}$, the offset adding counter $l_j \in \mathbb{N}$ and its transmission counter $c_j \in \mathbb{N}$.
The absolute value of the initial offset $u_j$ and the number of offset adding steps $L_j$ need to be chosen to be greater than the number of out-neighbors $\mathcal{D}_j^+$ of node $v_j$. 
Specifically, at initialization, each node $v_j$ chooses the number of steps $L_j$ and the offsets $u_j \in \mathbb{Z}_{<0}$, and $u_j[l_j] \in \mathbb{Z}_{\geq 0}$ for all $l_j \in \{ 0, 1, 2, ..., L_j \}$, to satisfy the following constraints:  
\begin{subequations}\label{Offset_value_1}
\begin{align}
%u_j < - \mathcal{D}_j^+, \label{Offset_value_1a} \\
L_j &\geq \mathcal{D}_j^+, \label{Offset_value_1c} \\
u_j &= - \sum_{l_j= 0}^{L_j} u_j[l_j], \label{Offset_value_1e} \\
u_j[l_j] &\geq 0, \ \forall \ l_j \in [0, L_j],\label{Offset_value_1b} \\
u_j[l_j] &= 0, \ \forall \ l_j \notin [0, L_j]. \label{Offset_value_1d} 
\end{align}
\end{subequations}
Constraints \eqref{Offset_value_1c}--\eqref{Offset_value_1d} are explicitly analyzed below: 
\begin{list4}
\item[1)] In \eqref{Offset_value_1c} the offset adding steps $L_j$ of every node $v_j$ need to be greater than or equal to node $v_j$'s out-degree so that every out-neighbor $v_i \in \mathcal{N}_j^+$ will receive at least one value of $u_j[l_j]$ from node $v_j$. 
As discussed in Section~\ref{subsec:conditions}, this is motivated by the privacy preservation guarantees.
\item[2)] Eq. \eqref{Offset_value_1e} means that the accumulated offset infused in the computation by node $v_j$ is equal to zero and the exact quantized average of the nodes' initial states can be calculated eventually without any error.
\item[3)] In \eqref{Offset_value_1b} the offset $u_j[l_j]$ which is injected to the network by each node $v_j$ each time its event-triggered conditions hold (i.e., for events $l_j \in [0, L_j]$) needs to be nonnegative so that (i) the event-triggered conditions (C1) and (C2) hold for every node after a finite number of steps and (ii) the exact quantized average of the initial states can be eventually calculated. 
\item[4)] Eq. \eqref{Offset_value_1d} means that node $v_j$ does not need to continue injecting nonzero offsets in the network so that exact quantized average of the initial states can be calculated without any error.
\end{list4}
The above choices imply that the initial offset $u_j$ every node $v_j$ injects in the network needs to be chosen so that, it is negative and satisfies $u_j \leq -\mathcal{D}_j^+$.
This is important to ensure that, during the operation of the proposed algorithm, the event-triggered conditions (C1) and (C2) hold for every node after a finite number of steps. 
If $u_j\geq 0$, the event-triggered conditions (C1) and (C2) may not hold and the proposed protocol may fail to calculate the average of the initial states.

\subsection{Algorithm Description}

The proposed algorithm is a quantized value transfer process in which every node in a strongly connected digraph $\mathcal{G}_d = (\mathcal{V}, \mathcal{E})$,  performs operations and transmissions according to a set of event-triggered conditions. 
The intuition behind the algorithm is as follows. Each node $v_j \in \mathcal{V}_p$ that would like to preserve its privacy performs the following steps: 
\begin{list4}
\item It initializes a counter $l_j$ to zero (i.e., $l_j=0$), and chooses the total number of offset adding steps $L_j$ such that $L_j \geq \mathcal{D}_j^+$ and the set of $(L_j+1)$ positive offsets $u_j[l_j]>0$ where $l_j \in \{0, 1, \ldots, L_j\}$.  
Finally it sets the initial negative offset $u_j$ that it injects to its initial state $y_j[0]$ to $u_j = - \sum_{l_j=0}^{L_j} u_j[l_j]$. 
For example, suppose that node $v_j$ has four out-neighbors. This means that it can choose $L_j = 6$, and then (randomly) set $u_j[0] = 1$, $u_j[1] = 3$, $u_j[2] = 2$, $u_j[3] = 4$, $u_j[4] = 1$, $u_j[5] = 2$, $u_j[6] = 5$; finally it sets $u_j = -18$.
\item It chooses an out-neighbor $v_l \in \mathcal{N}_j^+$ according to the unique order $P_{lj}$ (initially, it chooses $v_l \in \mathcal{N}_j^+$ such that $P_{lj}=0$) and transmits $z_j[0]$ and $\widetilde{y}_j[0] = y_j[0] + u_j$ to this out-neighbor. 
Then, it sets $\widetilde{y}_j[0] = 0$ and $z_j[0] = 0$.
\item During the execution of the algorithm, at every step $k$, node $v_j$ may receive a set of mass variables $\widetilde{y}_i[k]$ and $z_i[k]$ from each in-neighbor $v_i \in \mathcal{N}_j^-$. 
Then, node $v_j$ updates its state according to \eqref{subeq:1a}--\eqref{subeq:1b} (where in the sum of \eqref{subeq:1a} we use $\widetilde{y}_j[k]$) and checks whether any of its event-triggered conditions hold. 
If so, it injects an offset $u_j[l_j]$ to $y_j[k+1]$ and increases its offset increasing counter $l_j$ by one. 
Then, it sets its state variables $y^s_j[k+1]$ and $z^s_j[k+1]$ equal to $\widetilde{y}_j[k+1]=y_j[k+1]+u_j[l_j]$ and $z_j[k+1]$, respectively, and transmits them to an out-neighbor according to the predetermined order. 
If none of conditions \eqref{subeq:1a}--\eqref{subeq:1b} holds, node $v_j$ stores $y_j[k+1]$ and $z_j[k+1]$. 
If no message is received from any of the in-neighbors and no transmission takes place, the mass variables remain the same. 
\end{list4}

\noindent
The proposed algorithm is summarized in Algorithm~\ref{algorithm1}. 
Note here that curious nodes $v_c \in \mathcal{V}_c$ that try to identify the initial states of other nodes either follow Algorithm~\ref{algorithm1} or the underlying average consensus algorithm in Section~\ref{subsec:not_priv_alg}.

\begin{varalgorithm}{1}
\caption{Privacy-Preserving Event-Triggered Quantized Average Consensus with Event-Based Offset}
\noindent \textbf{Input:} A strongly connected digraph $\mathcal{G}_d = (\mathcal{V}, \mathcal{E})$ with $n=|\mathcal{V}|$ nodes and $m=|\mathcal{E}|$ edges. 
Each node $v_j\in \mathcal{V}$ has an initial state $y_j[0] \in \mathbb{Z}$. \\
\textbf{Initialization:} Each node $v_j \in \mathcal{V}_p$ does the following:
\begin{list4}
\item[1)] It assigns a \textit{unique order} $P_{lj}$ in the set $\{0,1,..., \mathcal{D}_j^+ -1\}$ to each of its out-neighbors $v_l \in \mathcal{N}^+_j$. 
\item[2)] It sets counter $c_j$ to $0$ and priority index $e_j$ to $c_j$.  
\item[3)] It sets counter $l_j$ to $0$, chooses $L_j \in \mathbb{N}$, where $L_j \geq \mathcal{D}_j^+$, and $u_j[k] \geq 0$ for $k \in \{0,1, \ldots, L_j\}$, and $u_j[k'] = 0$ for $k' > L_j$. 
It also sets $u_j = - \sum_{l_j=0}^{L_j} u_j[l_j]$. 
\item[4)] It sets $\widetilde{y}_j[0] = y_j[0] + u_j$, $z_j[0] = 1$, $z^s_j[0] = 1$ and $y^s_j[0] = \widetilde{y}_j[0]$ (which means that $q^s_j[0] = \widetilde{y}_j[0] / 1$). 
\item[5)] It selects out-neighbor $v_l \in \mathcal{N}_j^+$ such that $P_{lj} =e_j$ and transmits $z_j[0]$ and $\widetilde{y}_j[0]$ to this out-neighbor. Then, it sets $\widetilde{y}_j[0] = 0$ and $z_j[0] = 0$.
\item[6)] It sets $c_j = c_j + 1$ and $e_j = c_j \mod \mathcal{D}^+_j$.
\end{list4}
\textbf{Iteration:} For $k=0,1,2,\dots$, each node $v_j \in \mathcal{V}_p$, does the following:
\begin{list4}
\item \textbf{if} it receives $\widetilde{y}_i[k]$, $z_i[k]$ from at least one in-neighbor $v_i \in \mathcal{N}_j^-$ \textbf{then} it updates its values according to \eqref{subeq:1a}-\eqref{subeq:1b}.
\begin{list4a}
\item[$\circ$] \textbf{if} any of conditions (C1) and (C2) hold \textbf{then} 
\begin{list4a}
\item[-] it sets $\widetilde{y}_j[k+1] = u_j[l_j] + y_j[k+1]$ and $l_j \leftarrow l_j + 1$; 
\item[-] it sets $z^s_j[k+1] = z_j[k+1]$, $y^s_j[k+1] = \widetilde{y}_j[k+1]$ and $q^s_j[k+1] = \widetilde{y}_j[k+1] / z^s_j[k+1]$;
\item[-] it transmits $z_j[k+1]$ and $\widetilde{y}_j[k+1]$ to out-neighbor $v_{\lambda} \in \mathcal{N}_j^+$ for which $P_{\lambda j} = e_j$ and it sets $\widetilde{y}_j[k+1] = 0$, $y_j[k+1] = 0$ and $z_j[k+1] = 0$; 
\item[-] it sets $c_j = c_j + 1$ and $e_j = c_j \mod \mathcal{D}^+_j$. 
\end{list4a}
\item[$\circ$] \textbf{else} it stores $y_j[k+1]$ and $z_j[k+1]$.
\end{list4a}
\end{list4}
\textbf{Output:} \eqref{alpha_z_y} and \eqref{alpha_q} hold for every $v_j \in \mathcal{V}$. 
\label{algorithm1}
\end{varalgorithm}

\begin{remark}
Unlike other privacy preserving protocols proposed in the literature (see, e.g., \cite{Kefayati:2007, 2013:Nikolas_Hadj,Mo-Murray:2017, 2019:Allerton_themis}), the proposed strategy takes full advantage of the algorithm's finite time nature which means that consensus to the average of the initial states is reached after a finite number of iterations, while the error, introduced via the offset initially infused in the network by the nodes following the protocol, vanishes. 
\end{remark}

\vspace{-0.2cm}

\subsection{Convergence Analysis}

For the development of the necessary results regarding the operation of Algorithm~\ref{algorithm1} let us consider the following setup, the analysis of which for the non privacy-preserving case can be found in \cite{2020:Rikos_Mass_Accum}. 

{\it Setup:} Consider a strongly connected digraph $\mathcal{G}_d = (\mathcal{V}, \mathcal{E})$ with $n=|\mathcal{V}|$ nodes and $m=|\mathcal{E}|$ edges. 
During the execution of Algorithm~\ref{algorithm1}, at time step $k_0$, there is at least one node $v_{j'} \in \mathcal{V}$, for which 
\begin{equation}\label{great_z_prop1_det_1}
z_{j'}[k_0] \geq z_i[k_0], \ \forall v_i \in \mathcal{V}.
\end{equation}
Then, among the nodes $v_{j'}$ for which (\ref{great_z_prop1_det_1}) holds, there is at least one node $v_j$ for which 
\begin{equation}\label{great_z_prop2_det_2}
\widetilde{y}_j[k_0] \geq \widetilde{y}_{j'}[k_0] , \ v_j, v_{j'} \in \{ v_i \in \mathcal{V} \ | \ (\ref{great_z_prop1_det_1}) \ \text{holds} \}.
\end{equation}
For notational convenience we will call the mass variables of node $v_j$ for which (\ref{great_z_prop1_det_1}) and (\ref{great_z_prop2_det_2}) hold as the ``leading mass'' (or ``leading masses'').
Now we present the following two lemmas, which are helpful in the development of our results.

\begin{lemma}[\hspace{-0.00001cm}\cite{2020:Rikos_Mass_Accum}]\label{first_lemma}
Under the above {\it Setup}, the ``leading mass'' or ``leading masses'' at time step $k$, will always fulfill the ``Event-Trigger Conditions'' (C1) and (C2). 
This means that the mass variables of node $v_j$ for which (\ref{great_z_prop1_det_1}) and (\ref{great_z_prop2_det_2}) hold at time step $k_0$ will be transmitted (at time step $k_0$) by $v_j$ to an out-neighbor $v_l \in \mathcal{N}_j^+$. 
\end{lemma}

\begin{lemma}\label{second_lemma}
Under the above {\it Setup}, we have that if the event-triggered conditions of node $v_j \in \mathcal{V}_p$ are fulfilled in at least $(L_j+1)$ instances then, from \eqref{Offset_value_1c}--\eqref{Offset_value_1d}, we have that the accumulated amount of offset injected by node $v_j$ to the network becomes equal to zero. 
\end{lemma} 

\begin{proof}
Each node $v_j  \in \mathcal{V}_p$ at time step $k$ adds the offset $u_j[l_j]$ to its mass variable $y_j[k]$ if and only if the event-triggered conditions (C1) and (C2) hold. 
As a result, if the event-triggered conditions of $v_j$ are fulfilled for at least $(L_j+1)$ instances then the accumulated amount of offset node $v_j$ has injected in the computation becomes equal to zero. 
\end{proof}

%\begin{proof}
%Each node $v_j$ at time step $k$ adds the offset $u_j[l_j]$ to its mass variable $y_j[k]$ if and only if the event-triggered conditions (C1) and (C2) hold. 
%As a result, if the event-triggered conditions of $v_j$ are fulfilled for at least $(L_j+1)$ instances then the accumulated amount of offset node $v_j$ has injected in the computation becomes equal to zero. 
%\end{proof} 

The following theorem states that the proposed algorithm allows all nodes to reach quantized average consensus after a finite number of steps, for which we provide an upper bound.

\begin{theorem}\label{Privacy_Conver_Quant_Av}
Consider a strongly connected digraph $\mathcal{G}_d = (\mathcal{V}, \mathcal{E})$ with $n=|\mathcal{V}|$ nodes and $m=|\mathcal{E}|$ edges. 
The execution of Algorithm~\ref{algorithm1} allows each node $v_j \in \mathcal{V}$ to reach quantized average consensus after a finite number of steps, $\mathcal{S}_t$, bounded by {$\mathcal{S}_t \leq m^2(L_{\max}+1+n)$}, where $n$ is the number of nodes, $m$ is the number of edges in the network and $L_{\max} = \max_{v_j \in \mathcal{V}} L_j$ is the maximum value of offset adding steps chosen by nodes in the network. 
\end{theorem}

\begin{proof}
See Appendix~\ref{proof:theorem2}.
\end{proof}

\begin{remark}
Theorem~\ref{Privacy_Conver_Quant_Av} relies on the fact that if the proposed distributed protocol is executed for a finite number of time steps equal to $m^2 L_{\max}$, then every node in the network will receive a set of nonzero mass variables that are equal to the leading mass for at least $L_{\max}$ instances. 
This means that the event-triggered conditions will be fulfilled for each node in the network for at least $L_{\max}$ instances and, from Lemma~\ref{second_lemma}, the accumulated amount of offset injected in the network from each node is equal to zero. 
As a result, by executing the proposed protocol for an additional number of time steps equal to $n m^2$, we have that every nonzero mass in the network will merge to one leading mass (or multiple equally-valued leading masses) that is (are) equal to the average of the initial states, and subsequently this (these) leading mass (masses) will update the state variables of each node in the network, setting them equal to the average of the initial states. 
\end{remark}

\subsection{Topological conditions for privacy preservation} \label{subsec:conditions}

%\TC{Concerning observability, there does not exist any condition under which any two distinct states are distinguishable from quantized outputs \cite{MIYAMICHI:1982}.}

We establish topological conditions that ensure privacy for the nodes following the proposed protocol despite the presence of possibly colluding curious nodes in the network.

\begin{prop}\label{prop:1}
Consider a fixed strongly connected digraph $\mathcal{G}_d = (\mathcal{V}, \mathcal{E})$ with $n=|\mathcal{V}|$ nodes. Assume that a subset of nodes $\mathcal{V}_p$ follow the predefined privacy-preserving protocol, as described in Algorithm~\ref{algorithm1}, with offsets chosen as in \eqref{Offset_value_1c}-\eqref{Offset_value_1d}. 
%\todo{Let us define the set of curious nodes $\mathcal{V}_c \subset \mathcal{V}$. }
Curious node $v_c \in \mathcal{V}_c$ will not be able to identify the initial state $y_j [0]$ of $v_j \in \mathcal{V}_p$ , as long as  $v_j$ has 
\begin{list5}
\item[a)] at least one other node  (in- or out-neighbor) $v_{\ell}\in \mathcal{V}_p$ connected to it, or
\item[b)] has a non-curious in-neighbor $v_i \notin \mathcal{V}_p$ which first transmits to node $v_j$ at initialization.
\end{list5}
\end{prop} 
%For every node $v_j$ that wants to preserve the privacy of its initial value, if there exists at least one in-neighbor, say $v_{\ell}$, that follows the privacy-preserving protocol, the curious nodes will not be able to  infer $v_j$'s initial value. 
In other words, if the condition in Proposition~\ref{prop:1} is satisfied, the network will reach average consensus and the privacy of the initial states of the nodes following the privacy-preserving protocol will be preserved. 

\begin{proof}
See Appendix~\ref{proof:prop1}.
\end{proof}

Note here that a set of curious nodes could also attempt to ``estimate'' the initial states of some other nodes (e.g., by taking into account any available statistics about the initial states, $u_j$ and $L_j$). 
However, in our analysis in this paper, we are interested in whether the curious nodes can exactly infer the state of another node (see Definition~\ref{Definition_Quant_Privacy}). 
%The case where curious nodes attempt to ``estimate'' the initial states of other nodes will be considered as a future direction.

% ===============================================
%
%
% ALGORITHM 2
%
%
% ===============================================
\section{Initial zero-sum Offset Privacy-Preserving Strategy}\label{sec:distr_algo2}

In this section, we present and analyze another privacy preserving algorithm in which offsets are introduced only at the initialization stage. 
The main difference with the approach proposed in Section~\ref{sec:distr_algo} is that the proposed mechanism incorporates an offset to the mass variable of each node only during the initialization steps, in order to effectively preserve the privacy of its initial state. 

\subsection{Initialization for Quantized Privacy Strategy}

We have that each node maintains the variables $u_j \in \mathbb{Z}$ and $u^{(l)}_j \in \mathbb{Z}$ for every $v_l \in \mathcal{N}^+_j$ and its transmission counter $c_j \in \mathbb{N}$. 
Then, during initialization, it chooses the variables $u^{(l)}_j$ and $u_j$, to satisfy the following constraints. 
\begin{subequations}\label{Offset_value_2}
\begin{align}
u^{(l)}_j \in \mathbb{Z}, \ \forall v_l \in \mathcal{N}_j^+, \label{Offset_value_2a} \\
u_j = - \sum_{v_l \in \mathcal{N}_j^+} u^{(l)}_j, \label{Offset_value_2b}
\end{align}
\end{subequations}
Constraints \eqref{Offset_value_2a} and \eqref{Offset_value_2b} are explicitly analyzed below: 
\begin{list4}
\item[1)] In \eqref{Offset_value_2a} an integer offset for out-neighbor $v_l \in \mathcal{N}_j^+$, is selected in order to be infused to the initial state that node $v_j$ transmits to that out-neighbor $v_l \in \mathcal{N}_j^+$. 
The selection of a set of integer offsets, each one corresponding to an out-neighbor $v_l \in \mathcal{N}_j^+$, has to do with privacy preservation guarantees as discussed in Section~\ref{subsec:conditions_2}.
\item[2)] Eq. \eqref{Offset_value_2b} means that the accumulated offset infused in the network during the initialization steps by node $v_j$ is equal to zero and the exact quantized average of the initial states can be calculated without any error.
\end{list4}
It will be seen later that the strategy of incorporating an offset to the mass variable of each node only during the initialization steps increases the convergence speed of the proposed algorithm while it also improves the topological conditions that ensure privacy for the nodes following the proposed protocol.

\subsection{Algorithm Description}

The proposed algorithm is a quantized value transfer process in which each node in a strongly connected digraph $\mathcal{G}_d = (\mathcal{V}, \mathcal{E})$ performs operations and transmissions according to a set of event-triggered conditions. 
The intuition behind the algorithm is as follows. 
Each node $v_j$ that would like to preserve its privacy performs the following steps: 
\begin{list4}
\item It selects a set of integer offsets $u^{(l)}_j$, one for each out-neighbor $v_l \in \mathcal{N}_j^+$. 
It transmits the values $u^{(l)}_j$ to every out-neighbor $v_l \in \mathcal{N}_j^+$, while it receives the values $u^{(j)}_i$ from its in-neighbors $v_i \in \mathcal{N}_j^-$. 
Then, it sets its initial state equal to 
\begin{equation}\label{init_y_enh}
\widetilde{y}_j[0] = y_j[0] + u_j + \sum_{v_i \in \mathcal{N}_j^-}u^{(j)}_i ,
\end{equation}
where the initial offset $u_j$ is equal to $u_j = - \sum_{v_l \in \mathcal{N}_j^+} u^{(l)}_j$. 
For example, suppose that node $v_j$ has four out-neighbors ($\mathcal{D}_j^+ = 4$), three in-neighbors ($\mathcal{D}_j^- = 3$) and initial state $y_j[0] = 6$. 
This means that it can choose $u^{(0)}_j = 3$, $u^{(1)}_j = -2$, $u^{(2)}_j = 5$, and $u^{(3)}_j = -3$, while it sets $u_j = -3$. 
It transmits each of the values $u^{(0)}_j$, $u^{(1)}_j$, $u^{(2)}_j$, and $u^{(3)}_j$ to the corresponding out-neighbor, while it receives the values say, $8$, $3$ and $6$, from its in-neighbors. 
Then, from (\ref{init_y_enh}), it sets its initial state equal to $\widetilde{y}_j[0] = 20$. 
Note that (\ref{init_y_enh}) is essential not only for preserving the privacy of every node's initial quantized state but also to preserve the sum of the initial states, i.e., $\sum_{v_j \in \mathcal{V}} y_j[0] = \sum_{v_j \in \mathcal{V}} \widetilde{y}_j[0]$, as it will be seen later. 
\item It chooses an out-neighbor $v_l \in \mathcal{N}_j^+$ according to the unique order $P_{lj}$ (initially, it chooses $v_l \in \mathcal{N}_j^+$ such that $P_{lj}=0$) and transmits $\widetilde{y}_j[0]$ and $z_j[0]$ to this out-neighbor. Then, it sets $\widetilde{y}_j[0] = 0$ and $z_j[0] = 0$. 
\item During the execution of the algorithm, at every step $k$, node $v_j$ may receive a set of mass variables $\widetilde{y}_i[k]$ and $z_i[k]$ from each in-neighbor $v_i \in \mathcal{N}_j^-$. 
Then, node $v_j$ updates its state according to \eqref{subeq:1a}--\eqref{subeq:1b} and checks whether any of its event-triggered conditions hold (where in \eqref{subeq:1a} we use $\widetilde{y}_j[k]$). 
If so, it sets its state variables $y^s_j[k+1]$ and $z^s_j[k+1]$ equal to $\widetilde{y}_j[k+1]$ and $z_j[k+1]$, respectively, and transmits them to an out-neighbor according to the predetermined unique order. 
If none of the conditions \eqref{subeq:1a}--\eqref{subeq:1b} hold, then node $v_j$ stores $\widetilde{y}_j[k+1]$ and $z_j[k+1]$. 
Note that if no message is received from any of the in-neighbors, the mass variables remain the same. 
\end{list4}

\noindent
The proposed algorithm is summarized in Algorithm~\ref{algorithm2}. 
The curious nodes $v_c \in \mathcal{V}_c$ that try to identify the initial states of other nodes either follow Algorithm~\ref{algorithm2} or the underlying average consensus algorithm in Section~\ref{subsec:not_priv_alg}. 

\begin{varalgorithm}{2}
\caption{Privacy-Preserving Event-Triggered Quantized Average Consensus with Initial Zero-Sum Offset}
\noindent \textbf{Input:} A strongly connected digraph $\mathcal{G}_d = (\mathcal{V}, \mathcal{E})$ with $n=|\mathcal{V}|$ nodes and $m=|\mathcal{E}|$ edges. 
Each node $v_j\in \mathcal{V}$ has an initial state $y_j[0] \in \mathbb{Z}$. \\
\textbf{Initialization:} Each node $v_j \in \mathcal{V}_p$ does the following:
\begin{list4}
\item[1)] It assigns a \textit{unique order} $P_{lj}$ in the set $\{0,1,..., \mathcal{D}_j^+ -1\}$ to each of its out-neighbors $v_l \in \mathcal{N}^+_j$. 
\item[2)] It sets counter $c_j$ to $0$ and priority index $e_j$ to $c_j$.  
\item[3)] It chooses $u^{(l)}_j \in \mathbb{Z}$, for every $v_l \in \mathcal{N}_j^+$. 
Then, it sets $u_j = - \sum_{v_l \in \mathcal{N}_j^+} u^{(l)}_j$.
\item[4)] It transmits $u^{(l)}_j$ to each $v_l \in \mathcal{N}_j^+$.
\item[5)] It sets $\widetilde{y}_j[0] = y_j[0] + u_j + \sum_{v_i \in \mathcal{N}_j^-} u^{(j)}_i$, $z_j[0] = 1$, $y^s_j[0] = \widetilde{y}_j[0]$ and $z^s_j[0] = z_j[0]$ (which means that $q^s_j[0] = \widetilde{y}_j[0] / z^s_j[0]$).
\item[6)] It selects out-neighbor $v_l \in \mathcal{N}_j^+$ such that $P_{lj} =e_j$ and transmits $\widetilde{y}_j[0]$ and $z_j[0]$ to this out-neighbor. 
Then, it sets $\widetilde{y}_j[0] = 0$ and $z_j[0] = 0$. 
\item[7)] It sets $c_j = c_j + 1$ and $e_j = c_j \mod \mathcal{D}^+_j$. 
\end{list4}
\textbf{Iteration:} For $k=0,1,2,\dots$, each node $v_j \in \mathcal{V}_p$ does the following:
\begin{list4}
\item \textbf{if} it receives $\widetilde{y}_i[k]$ and $z_i[k]$ from at least one in-neighbor $v_i \in \mathcal{N}_j^-$ \textbf{then} it updates $\widetilde{y}_j[k+1]$ and $z_j[k+1]$ according to \eqref{subeq:1a}-\eqref{subeq:1b}.
\begin{list4a}
\item[$\circ$] \textbf{if} any of conditions (C1) and (C2) hold \textbf{then} 
\begin{list4a}
\item[-] it sets $y^s_j[k+1] = \widetilde{y}_j[k+1]$, $z^s_j[k+1] = z_j[k+1]$ and $q^s_j[k+1] = y^s_j[k+1] / z^s_j[k+1]$;
\item[-] it transmits $\widetilde{y}_j[k+1]$ and $z_j[k+1]$ to out-neighbor $v_{\lambda} \in \mathcal{N}_j^+$ for which $P_{\lambda j} = e_j$ and it sets $\widetilde{y}_j[k+1] = 0$ and $z_j[k+1] = 0$; 
\item[-] it sets $c_j = c_j + 1$ and $e_j = c_j \mod \mathcal{D}^+_j$. 
\end{list4a}
\item[$\circ$] \textbf{else} it stores $\widetilde{y}_j[k+1]$ and $z_j[k+1]$.
\end{list4a}
\end{list4}
\textbf{Output:} \eqref{alpha_z_y} and \eqref{alpha_q} hold for every $v_j \in \mathcal{V}$. 
\label{algorithm2}
\end{varalgorithm}

\subsection{Deterministic Convergence Analysis}

The following theorem states that Algorithm~\ref{algorithm2} allows all nodes to reach quantized average consensus after a finite number of steps, for which we provide an upper bound.

\begin{theorem}\label{Privacy_Conver_Quant_Av_Alg2}
Consider a strongly connected digraph $\mathcal{G}_d = (\mathcal{V}, \mathcal{E})$ with $n=|\mathcal{V}|$ nodes and $m=|\mathcal{E}|$ edges. 
The execution of Algorithm~\ref{algorithm2} allows each node $v_j \in \mathcal{V}$ to reach quantized average consensus after a finite number of steps, $\mathcal{S}_t$, bounded by {$\mathcal{S}_t \leq m^2 n$}, where $n$ is the number of nodes and $m$ is the number of edges in the network. 
\end{theorem}

\begin{proof}
See Appendix~\ref{proof:theorem3}.
\end{proof}

\vspace{-.5cm}

\subsection{Topological conditions for privacy preservation} \label{subsec:conditions_2}

The topological conditions for ensuring privacy for the nodes following Algorithm~\ref{algorithm2} are given in Proposition~\ref{prop_privacy_2}.

\begin{prop}\label{prop_privacy_2}
Consider a fixed strongly connected digraph $\mathcal{G}_d = (\mathcal{V}, \mathcal{E})$ with $n=|\mathcal{V}|$ nodes. 
Assume that a subset of nodes $\mathcal{V}_p$ follow the predefined privacy-preserving protocol, as described in Algorithm~\ref{algorithm2}, with offsets chosen as in \eqref{Offset_value_2a}-\eqref{Offset_value_2b}. 
%\todo{Let us define the set of curious nodes $\mathcal{V}_c \subset \mathcal{V}$. }
Curious node $v_c \in \mathcal{V}_c$ will not be able to identify the initial state $y_j[0]$ of $v_j$, as long as $v_j \in \mathcal{V}_p$ and there exists at least one out-neighbor $v_l$ that is not curious (i.e., $v_l \notin \mathcal{V}_c$). 
\end{prop} 
%For every node $v_j$ that wants to preserve the privacy of its initial value, if there exists at least one in-neighbor, say $v_{\ell}$, that follows the privacy-preserving protocol, the curious nodes will not be able to  infer $v_j$'s initial value. 

\begin{proof}
%Let us assume that node $v_j$ follows the privacy-preserving protocol, as described in Algorithm~\ref{algorithm2}. 
%We consider the following simple scenarios, which constitute the building blocks of the directed network, due to the token-based approach of the privacy-preserving protocol. 
As it was the case with Algorithm~\ref{algorithm1}, the topological conditions will be extracted from simple scenarios, which constitute the building blocks of the directed network.
\begin{list5}
\item[1)] It is easy to observe that if all the in- and out-neighbors of node $v_j$ are curious and they communicate with each other, it is not possible for this node to keep its privacy. 
At initialization, the curious nodes will know the values $v_j$ transmitted to its out-neighbors, i.e., $u^{(l)}_j$ to every $v_l \in \mathcal{N}_j^+$. 
Also, the curious nodes will know the values $v_j$ received during the initialization, i.e., $u^{(j)}_i$ from every $v_i \in \mathcal{N}_j^-$. 
Hence, after the Initialization of Algorithm~\ref{algorithm2}, the curious nodes will compute the initial offset of node $v_j$, since the initial offset satisfies \eqref{Offset_value_2b}; hence, privacy of $v_j$'s initial state will not be preserved. 
As a result, at least one neighbor that is not curious is needed.
\item[2)] Let us consider the case for which there exists at least one out-neighbor of node $v_j$, say $v_l$, that is neither curious nor following the privacy-preserving protocol, and all other in- and out-neighbors of both nodes are curious. 
During the Initialization of Algorithm~\ref{algorithm2}, $v_l$ will receive the value $u^{(l)}_j$ from $v_j$ and will sum it with its own initial state in order to calculate $\widetilde{y}_l[0]$. 
Then, $v_j$ will calculate its own $\widetilde{y}_j[0]$ and as a result the privacy of nodes $v_j$ and $v_l$ is preserved. 
\end{list5}
From these discussions, we can deduce that, in order for node $v_j$ to preserve the privacy of its initial state, it is sufficient that $v_j \in \mathcal{V}_p$ and there exists at least one non-curious out-neighbor $v_l \notin \mathcal{V}_c$. 
Note that in such cases, the initial states of both nodes are protected (though the sum of these states may be exposed). 
As a result, the initial state of the node that follows the protocol cannot be inferred exactly, which according to Definition~\ref{Definition_Quant_Privacy} implies that its privacy is preserved.
\end{proof}

\begin{remark}
Note that the topological conditions of Algorithm~\ref{algorithm2} are improved with respect to those of Algorithm~\ref{algorithm1}. 
While the topological conditions of Algorithm~\ref{algorithm1} require two connected nodes to follow the privacy-preserving protocol or an in-neighbor to explicitly transmit to the privacy-preserving node during its initialization, the topological conditions of Algorithm~\ref{algorithm2} do not necessarily need a node that follows the protocol or has side information. 
Also, apart from improving the topological conditions that ensure privacy for the nodes following the proposed protocol, Algorithm~\ref{algorithm2} also exhibits increased convergence speed (seen the subsequent discussion in Section~\ref{sec:results}), since the injection of the zero-valued offset requires only one time step during the initialization procedure.
\end{remark}

% ===============================================
%
%
% SIMULATIONS
%
%
% ===============================================

\section{Simulation Results} \label{sec:results}

In this section, we present simulation results to illustrate the behavior of our proposed distributed protocols. 
Specifically, we analyze the cases of: 
\begin{itemize}
\item[A)] a {randomly generated} digraph of $20$ nodes with the average of the initial states of the nodes turning out to be equal to $q = {181}/{20} = 9.05$, 
\item[B)] $1000$ randomly generated digraphs of $20$ nodes each where, for convenience in plotting the results, the initial quantized state of each node remained the same (for each one of the $1000$ randomly generated digraphs); this means that the average of the nodes' initial quantized states also remained equal to $q = {185}/{20} = 9.25$. 
\end{itemize}
For each of the above cases we analyze the scenarios where each node $v_j \in \mathcal{V}$ does the following: Case (i) executes the privacy protocol described in Algorithm~\ref{algorithm1} and initially infuses in the network a {randomly chosen} offset $u_j \in [-100, -50]$ with {randomly chosen} offset adding steps $L_j \in [20, 40]$, Case (ii) executes the privacy protocol described in Algorithm~\ref{algorithm2} and initially infuses in the network the randomly chosen offset $u_j \in [-100, 100]$, the randomly chosen offsets $u^{(l)}_j \in [-20, 20]$, for every $v_l \in \mathcal{N}_j^+$, and $u_j$ chosen according to \eqref{Offset_value_2b} and Case (iii) initially does not infuse any offset in the network, i.e., $u_j = 0$ and $u^{(l)}_j= 0$ for every $v_l \in \mathcal{N}_j^+$, which means that it does not attempt to preserve the privacy of its initial state. 
Note that the digraphs were randomly generated by creating, independently for each ordered pair $(v_j,v_i)$ of two nodes $v_j$ and $v_i$ ($v_j \neq v_i$), a directed edge from node $v_i$ to node $v_j$ with probability $p = 0.3$.

\subsection{Execution of Algorithm~\ref{algorithm1} and Algorithm~\ref{algorithm2} over a Random Digraph of $20$ Nodes}

In Fig.~\ref{Journal_9_05_20Rand1Graph}, we illustrate Algorithm~\ref{algorithm1} and Algorithm~\ref{algorithm2} over a random digraph of $20$ nodes, where the average of the initial states of the nodes is equal to $q = {181}/{20} = 9.05$. 
We analyze the operation of our algorithms for the scenarios where each node $v_j \in \mathcal{V}$:	Case (i) executes Algorithm~\ref{algorithm1} and initially infuses in the network an offset $u_j \in [-100, -50]$ with offset adding steps $L_j \in [20, 40]$ (see top of Fig.~\ref{Journal_9_05_20Rand1Graph}), Case (ii) executes Algorithm~\ref{algorithm2} and initially infuses in the network the randomly chosen offset $u_j \in [-100, 100]$ and the offsets $u^{(l)}_j \in [-20, 20]$, for every $v_l \in \mathcal{N}_j^+$ (see middle of Fig.~\ref{Journal_9_05_20Rand1Graph}), and Case (iii) initially does not infuse any offset in the network $u_j = 0$, $u^{(l)}_j = 0$, for every $v_l \in \mathcal{N}_j^+$ (see bottom of Fig.~\ref{Journal_9_05_20Rand1Graph}). 
Here, we observe that for Case (iii) (where each node $v_j$ does not attempt to preserve the privacy of its initial state) each node is able to calculate the average of the initial states after $100$ time steps. 
However, for the case where each node $v_j$ wants to preserve the privacy of its initial state we observe that Algorithm~\ref{algorithm1} converges after $450$ time steps, while Algorithm~\ref{algorithm2} converges after $105$ time steps. 
Furthermore, we observe that for Case (iii) we have $y_j[0] \in [3, 19]$ for every $v_j \in \mathcal{V}$, for the case where every node executes Algorithm~\ref{algorithm1} we have $\widetilde{y}_j[0] \in [-40, -80]$ for every $v_j \in \mathcal{V}$, and for the case where every node executes Algorithm~\ref{algorithm2} we have $\widetilde{y}_j[0] \in [-80, 90]$ for every $v_j \in \mathcal{V}$.
Note here that both algorithms are able to calculate the \textit{exact} average of the initial states of the nodes without introducing any error due to the utilized privacy-preserving strategy.

\begin{figure}[t]
\begin{center}
\includegraphics[width=.8\columnwidth]{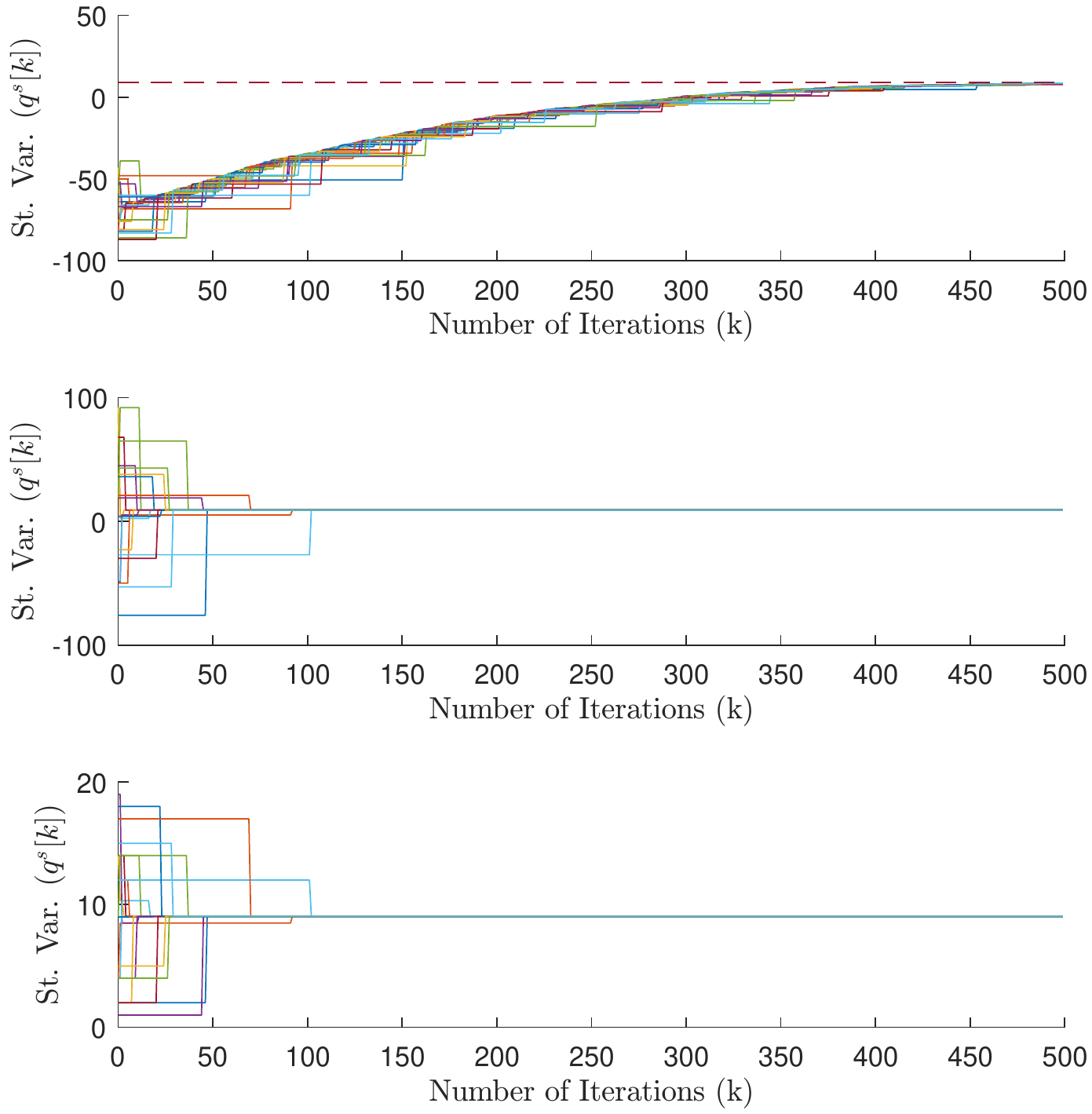}
\caption{Execution of Algorithm~\ref{algorithm1} and Algorithm~\ref{algorithm2} for a random digraph of $20$ nodes. 
\textit{Top figure:} Node state variables with privacy preservation of Algorithm~\ref{algorithm1} plotted against the number of iterations where the dashed line is the average of the initial states.
\textit{Middle figure:} Node state variables with privacy preservation of Algorithm~\ref{algorithm2} plotted against the number of iterations.
\textit{Bottom Figure:} Node state variables without privacy preservation plotted against the number of iterations.\vspace{-0.2cm}}
\label{Journal_9_05_20Rand1Graph}
\end{center}
\end{figure}

\subsection{Execution of Algorithm~\ref{algorithm1} and Algorithm~\ref{algorithm2} Averaged over $1000$ Random Digraphs of $20$ Nodes}

In Fig.~\ref{Journal_9_25_1000Graphs_20Nodes} we present the same cases as in Fig.~\ref{Journal_9_05_20Rand1Graph} with the difference being that they are averaged over $1000$ randomly generated digraphs of $20$ nodes. 
The initial quantized state of each node remained the same for each one of the $1000$ randomly generated digraphs (in particular the average of the initial states of the nodes is equal to $q = {185}/{20} = 9.25$). 
For every node $v_j$, the initial offset $u_j \in [-100, -50]$ and the offset adding steps $L_j \in [20, 40]$ during the execution of Algorithm~\ref{algorithm1}, as well as the initial offset $u_j \in [-100, 100]$ and the offsets $u^{(l)}_j \in [-20, 20]$, for every $v_l \in \mathcal{N}_j^+$, during the execution of Algorithm~\ref{algorithm2}, were randomly chosen for each digraph according to a uniform distribution. 
We can see that the main results resemble those in Fig.~\ref{Journal_9_05_20Rand1Graph}, and Algorithm~\ref{algorithm1} converges after $450$ time steps while Algorithm~\ref{algorithm2} converges after $170$ time steps. 
For the Case (iii) we have $y_j[0] \in [4, 19]$ for every $v_j \in \mathcal{V}$, whereas for the case where every node executes Algorithm~\ref{algorithm1} we have $\widetilde{y}_j[0] \in [-40, -70]$ for every $v_j \in \mathcal{V}$, and for the case where every node executes Algorithm~\ref{algorithm2} we have $\widetilde{y}_j[0] \in [3, 20]$ for every $v_j \in \mathcal{V}$.
This means that, for Case (iii), the states $y_j[0]$ of every $v_j \in \mathcal{V}$ are almost equal to the states $\widetilde{y}_j[0]$ for the case where every node executes Algorithm~\ref{algorithm2}, since the offsets $u_j \in [-100, 100]$, $u^{(l)}_j \in [-20, 20]$ for every $v_l \in \mathcal{N}_j^+$ were randomly chosen for each digraph according to a uniform distribution.

\begin{figure}[t]
\begin{center}
\includegraphics[width=.85\columnwidth]{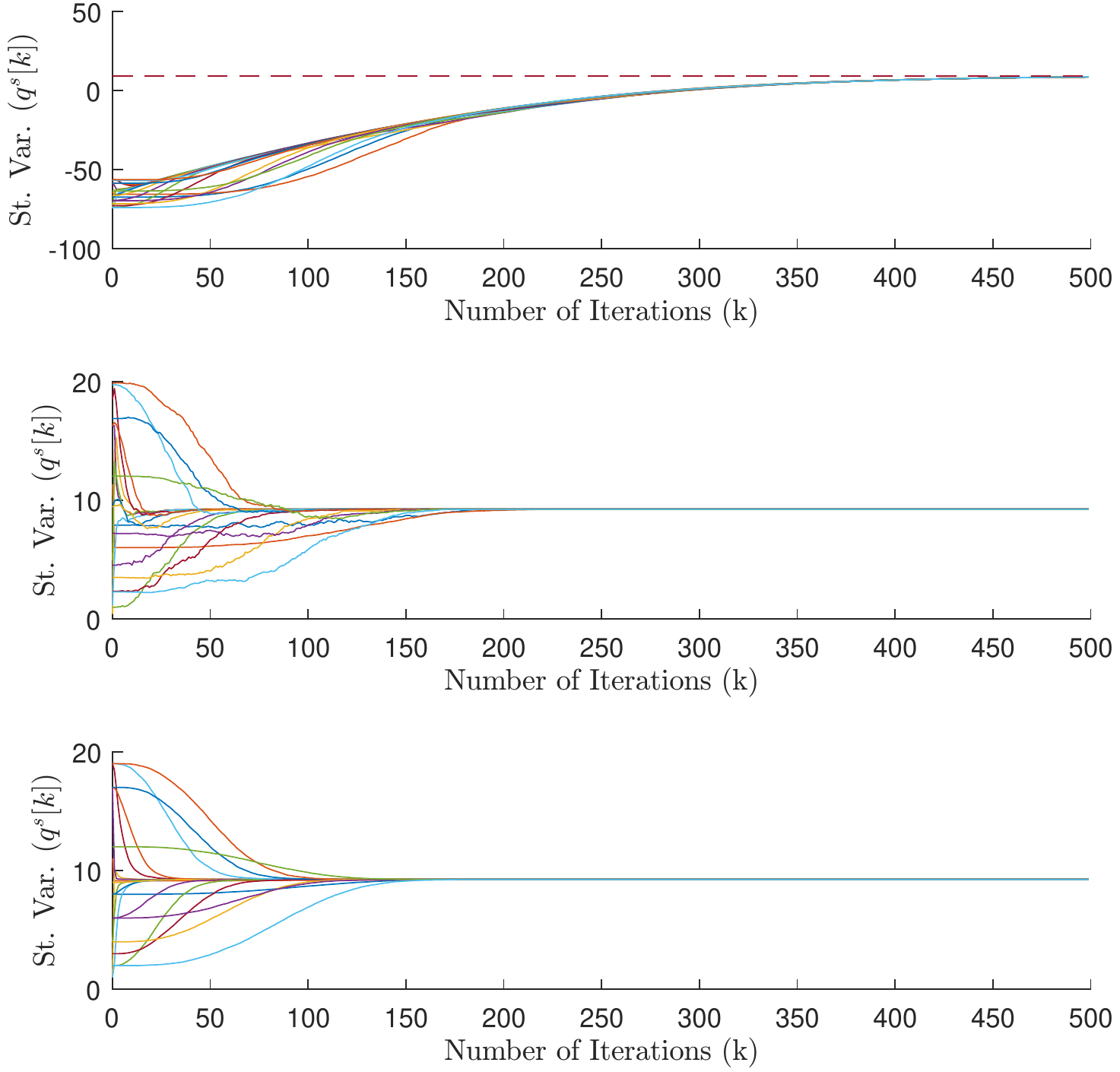}
\caption{Execution of Algorithm~\ref{algorithm1} and Algorithm~\ref{algorithm2} averaged over $1000$ random digraphs of $20$ nodes.  
\textit{Top figure:} Average values of node state variables with privacy preservation of Algorithm~\ref{algorithm1} plotted against the number of iterations (averaged over $1000$ random digraphs of $20$ nodes) where the dashed line is the average of the initial states.
\textit{Middle figure:} Average values of node state variables with privacy preservation of Algorithm~\ref{algorithm2} plotted against the number of iterations (averaged over $1000$ random digraphs of $20$ nodes). 
\textit{Bottom Figure:} Average values of node state variables without privacy preservation plotted against the number of iterations (averaged over $1000$ random digraphs of $20$ nodes).\vspace{-0.2cm}}
\label{Journal_9_25_1000Graphs_20Nodes}
\end{center}
\end{figure}

\begin{remark}
In Fig.~\ref{Journal_9_05_20Rand1Graph} and Fig.~\ref{Journal_9_25_1000Graphs_20Nodes} we can see that both algorithms are able to calculate the exact average of the initial states of the nodes without introducing any error due to the utilized privacy preserving strategies. 
This makes Algorithm~\ref{algorithm1} and Algorithm~\ref{algorithm2} the first algorithms in the literature which calculate the exact average of the initial states of the nodes in finite time without introducing any error in a privacy-preserving manner.
Furthermore, the privacy strategy presented in Algorithm~\ref{algorithm2} requires fewer time steps for convergence than the strategy in Algorithm~\ref{algorithm1} since the injection of the zero-valued offset is done during the initialization procedure and requires only one time step. 
However, Algorithm~\ref{algorithm1} allows a more efficient usage of the available network resources (e.g., communication bandwidth) since each node is required to transmit to at most one out-neighbor at each time step $k$. 
\end{remark}

% an application for calculating, under privacy-preserving guarantees, the average number of people infected from an infectious disease (that also imposes constraints on movement and interconnectivity, such as COVID-$19$)

\vspace{-.5cm}

% ===============================================
%
%
% CONCLUSIONS
%
%
% ===============================================

\section{Conclusions and Future Directions}\label{sec:conclusions}

%\subsection{Conclusions}

In this paper, we proposed two event-triggered quantized privacy-preserving strategies which allow the nodes of a multi-agent system to calculate the average of their initial states using quantized states and after a finite number of time steps without revealing their initial state to other nodes. 
%Specifically, our first protocol consists of initially adding a negative quantized offset (determined by the node and the number of times certain conditions are satisfied) and then, for a predetermined number of events each node injects a set of positive offsets in the network such that the total offset is canceled out. 
%Our second protocol consists of adding a zero-valued offset only during the algorithm's initialization procedure in which each node transmits quantized values to every out-neighbor. 
%We combine our proposed privacy preserving protocols with our quantized averaging algorithm in \cite{2020:Rikos_Mass_Accum}, and we show that 
They take full advantage of the algorithm's finite time nature which means that consensus to the \textit{exact} average of the initial states is achieved after a finite number of iterations which we explicitly calculated, while the error, introduced from the offset initially infused in the network by the nodes following the protocol, vanishes completely. % (a characteristic not present in \cite{2013:Nikolas_Hadj, Mo-Murray:2017}). 
Finally, we have demonstrated the performance of our proposed protocols via illustrative examples and we have presented an application in smart grids. %for calculating in a privacy-preserving fashion the average number of infected persons under constraints in their movement and interconnectivity. 

%\vspace{-.5cm}

%\subsection{Future Directions}

The point-to-point communication protocol and the quantized nature of the packets used in the proposed algorithms facilitate the use of cryptographic primitives for setting up secure channels and preventing eavesdropping, while harvesting the benefits of event-triggered and finite-time operation of the distributed privacy-preserving protocol proposed. 
We plan to deploy such cryptographic protocols and implement our algorithms in practice. 
%Furthermore, a privacy preserving strategy which combines the main aspects of both proposed protocols may lead to a more efficient usage of network resources while maintaining fast convergence speed. 
%These directions we plan to exploit in future work. 

% ------------------------------------------------------------------------------
% Bibliography
% ------------------------------------------------------------------------------
\bibliographystyle{IEEEtran}
\bibliography{bibliografia_consensus}

%\vspace{.5cm}
\appendices
\section{Application: Power Request in Smart Grids under Privacy-Preserving Guarantees}\label{app_analysis}

In this application, by applying either Algorithm~\ref{algorithm1} or Algorithm~\ref{algorithm2} in the network shown in Fig.~\ref{app_fig}, we are able compute distributively and in a privacy-preserving manner the \emph{total} power requested within a certain time period by a set of interconnected nodes (be it households, electric cars, etc.) in a neighborhood. 
To formally define the privacy-preserving total power request of a set of interconnected nodes in a neighborhood, let $\mathcal{B} = \{ b_1, b_2, \ldots, b_n \}$ denote the set of $n$ interconnected nodes in neighborhood $\mathcal{B}$, and $\mathcal{C}^{\day} = \{ c_1^{\day}, c_2^{\day}, \ldots, c_n^{\day} \}$, denote the set of requested powers per household at each day within a month, where $\day$ refers to the day of the month. 
This means that for node $b_j$ the amount of requested power at the third day of a month is denoted as $c_j^{3}$. For simplicity of exposition, we show how our algorithms work within a single day, so hereafter we drop the index $\day$.
%\TC{Our aim is to compute distributively the \emph{total} requested power by a set of $n$ interconnected nodes in a neighborhood every day. For doing so, it suffices to compute the average among them and be aware of the total number of nodes participating.}

% we are able to calculate the total requested power of a set of $n$ households over a neighborhood within a specific day (i.e., we are able to calculate $\overline{c}^{\day} \coloneqq \bigl( \sum_{j=1}^n c_j^{\day} \bigr ) / n$), while at the same time we preserve the privacy of information regarding the residents' daily routine. 
%Specifically, let us suppose that during the first day of a month household $b_i$ requests $5$ $kW$ i.e., $c_i^{1}[0] = 5$, while every other household $b_j \in \mathcal{B} \setminus \{b_i\}$ requests $c_j[0] \in [ 15, 20 ]$. 
Both algorithms initially use as input the requested power of each node $b_j \in \mathcal{B}$ for a specific day $c_j$ and create a distorted version $\widetilde{c}_j$.  Let $\widetilde{\mathcal{C}} = \{ \widetilde{c}_1, \widetilde{c}_2, \ldots, \widetilde{c}_n \}$, be the set of distorted amounts of requested power from every household in the neighborhood at each day within a month. 
Eventually, the smart meter collects the calculated average demand of the neighborhood, in order to multiply it with the number of participating nodes (there exist algorithms for computing the total number of nodes, in case it can vary e.g., due to the presence of excess electric vehicles) and calculate the total demanded power.
The operation of both algorithms is the following: 

\noindent 
A. During the operation of Algorithm~\ref{algorithm1} a set of offset adding steps $L_j$ is chosen from each node $b_j$.  
Then, each node injects a set of positive offsets for a number of $L_{\max} = \max_{b_j \in \mathcal{B}} L_j$ steps, which will guarantee the preservation of the privacy of the people living in this household. 
%Finally, at the last day of each month (i.e., $day = 30$) we have $\widetilde{c}^{\day}_j = c^{\day}_j +  \sum_{day=1}^{29} \widetilde{c}^{\day}_j$ and $L_j=0$ for every household $b_j$, which is important for receiving the correct billing information at the end of each month. 

\noindent
B. During the operation of Algorithm~\ref{algorithm2}, each node $b_j$ transmits nonzero offsets to its out-neighbors. 
Then, for each day the initial states $\widetilde{c}_j$ of every household $b_j$ are calculated (note that it holds $\sum_{j=1}^{n} \widetilde{c}_j = \sum_{j=1}^{n} c_j$). 
This means that the privacy of the data containing the daily requested power of each household is preserved. 
%Furthermore, each household $b_j \in \mathcal{B}$ ensures that $\sum_{\day=1}^{30} \widetilde{c}^{\day}_j = \sum_{\day=1}^{30} c^{\day}_j$ which is important for receiving the correct billing information at the end of each month. 

For the households in Fig.~\ref{app_fig} let us consider the set $\mathcal{C} = \{ 30, 35, 28, 34, 27, 37, 29, 32 \}$, which denotes the amount of requested power at a certain day from each household (i.e., household $b_1$ requests $30$, $b_2$ requests $35$, etc.) with the average demand being $\overline{c}^{\day} = 31.5$. 
During the execution of Algorithm~\ref{algorithm1} we have 
%$u_j = \{-15, -19, -13, -17, -12, -20, -14, -16\}$ which means that 
$\widetilde{\mathcal{C}} = \{ 15, 16, 15, 17, 15, 17, 15, 16 \}$, while during the execution of Algorithm~\ref{algorithm2} we have 
%$u_1^{2} = 7$, $u_1^{4} = 3$, $u_2^{1} = 5$, $u_2^{5} = 9$, $u_3^{4} = 4$, $u_3^{7} = 8$, $u_4^{1} = 3$, $u_4^{3} = 5$, $u_4^{5} = 4$, $u_4^{7} = 6$, $u_5^{2} = 2$, $u_5^{4} = 3$, $u_5^{6} = 4$, $u_5^{8} = 6$, $u_6^{5} = 4$, $u_6^{8} = 8$, $u_7^{3} = 4$, $u_7^{4} = 6$, $u_8^{5} = 7$, $u_8^{6} = 5$, which means 
$\widetilde{\mathcal{C}} = \{ 28, 30, 25, 32, 36, 34, 33, 34 \}$. 
%The $\widetilde{\mathcal{C}}^{\day}$ and $C^{\day}$ states of Algorithm~\ref{algorithm1} and Algorithm~\ref{algorithm2} are plotted in Fig.~\ref{app_plot2} for $day=1$. 
Note here that the daily demands $\widetilde{\mathcal{C}}^{\day}$ for Algorithm~\ref{algorithm1} are different than the daily demands $\widetilde{\mathcal{C}}^{\day}$ for Algorithm~\ref{algorithm2}, due to the choice of $u^{\day}_j$ and $u_j^{l}$, $\forall v_l \in \mathcal{N}^+_j$, respectively. 
In Fig.~\ref{app_plot2} we can see that this choice affects the convergence rate of both algorithms with Algorithm~\ref{algorithm2} generally requiring less time steps to converge compared to Algorithm~\ref{algorithm1} (as already discussed in Section~\ref{sec:results}). 
However, note that from Fig.~\ref{Journal_9_25_1000Graphs_20Nodes} we can see that Algorithm~\ref{algorithm2}, compared to Algorithm~\ref{algorithm1}, requires less time steps to converge also for the average case of $1000$ random digraphs of $20$ nodes due to its privacy strategy which is implemented only during the initialization steps. 
%Specifically, during Algorithm~\ref{algorithm1}, each node initially incorporates a negative offset and then a set of nonzero offsets during the iteration steps, while, during Algorithm~\ref{algorithm2}, each node incorporates a set of offsets only during the initialization steps; for this reason, we have that the initial conditions for both algorithms will in general be different for the same case, which can also be seen in Fig.~\ref{app_plot2}. 

\begin{figure}[t]
\begin{center}
\includegraphics[width=.85\columnwidth]{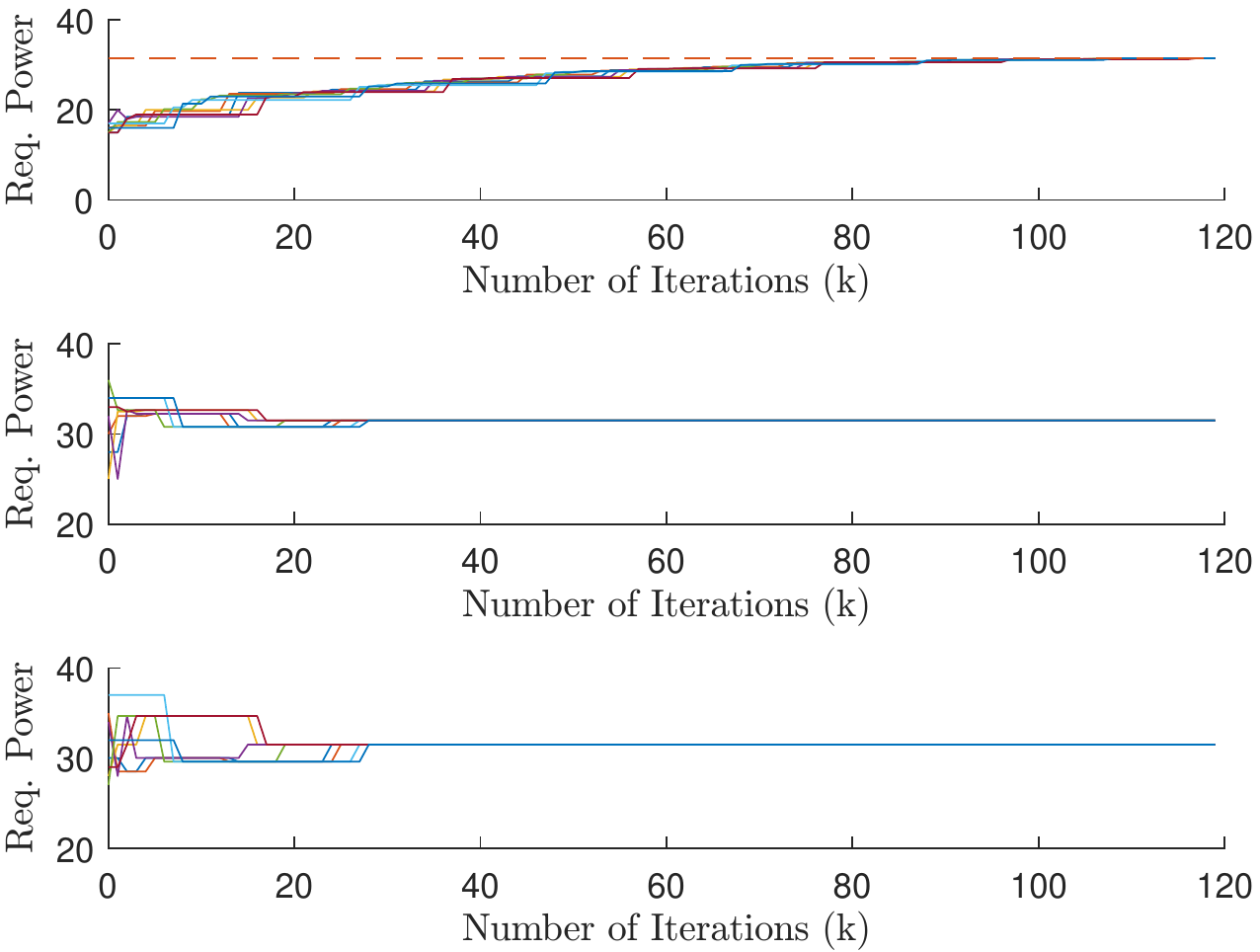}
\caption{Execution of Algorithm~\ref{algorithm1} and Algorithm~\ref{algorithm2} for the neighborhood of $8$ households shown in Fig.~\ref{app_fig}. 
\textit{Top figure:} Requested power per household with privacy preservation of Algorithm~\ref{algorithm1} plotted against the number of iterations for $day=1$, where the dashed line is the average of the initial states.
\textit{Middle figure:} Requested power per household with privacy preservation of Algorithm~\ref{algorithm2} plotted against the number of iterations for $day=1$. 
\textit{Bottom Figure:} Requested power per household without privacy preservation plotted against the number of iterations for $day=1$.\vspace{-0.2cm}}
\label{app_plot2}
\end{center}
\end{figure}

%In Fig.~\ref{app_plot2} both algorithms preserve the privacy of the requested power of each household in the neighborhood for $day=1$ while they are able to preserve the privacy of each household also for the rest of the month's days (i.e., $day=2, ..., 30$). As a result, the proposed algorithms preserve the privacy of the requested power of each household while they guarantee that each household receives the correct billing information at the end of each month. 

\vspace{-.3cm}

% ===============================================
%
% Theorem 2
%
% ===============================================
\section{Proof of Theorem~\ref{Privacy_Conver_Quant_Av}} \label{proof:theorem2}

Let us assume that, at time step $k_0$, the mass variables of node $v_j$ are the ``leading mass'' and there exists a set of nodes $\mathcal{V}^f [k_0] \subseteq \mathcal{V}$ which is defined as $\mathcal{V}^f [k_0] = \{ v_i \in \mathcal{V} \ | \ z_i[k_0] > 0 \ \text{but} \ (\ref{great_z_prop1_det_1}) \ \text{or} \ (\ref{great_z_prop2_det_2}) \ \text{do not hold} \}$ (i.e., it is the set of nodes which have nonzero mass variables at time step $k_0$ but they are not ``leading masses'').
Note here that if the ``leading mass'' reaches a node simultaneously with some other (leading or otherwise) mass then it gets ``merged'', i.e., the receiving node ``merges'' the mass variables it receives, by summing their numerators and their denominators, creating a set of mass variables with a greater denominator. 
Furthermore, we will say that the ``leading mass'', gets ``obstructed'' if it reaches a node whose state variables are greater than the mass variables (i.e., either the denominator of the node's state variables is greater than the denominator of the mass variables, or, if the denominators are equal, the numerator of the state variables is greater than the numerator of the mass variables). 
Note that if the ``leading mass'' we started off with at time step $k_0$ gets ``obstructed'' then its no longer the ``leading mass'', since it will not fulfill the event-triggered conditions of the corresponding node (from Lemma~\ref{first_lemma} we have that the ``leading mass'' always fulfills the event-triggered conditions). 

Suppose that the ``leading mass'' at time step $k_0$ is held by node $v_j$ and is given by $\widetilde{y}_j[k_0]$ and $z_j[k_0]$. 
Since it does not get merged or obstructed, during the execution of Algorithm~\ref{algorithm1}, it will reach every node $v_j \in \mathcal{V}$ in at most $m^2$ steps, where $m = | \mathcal{E} |$ is the number of edges of the given digraph $\mathcal{G}_d$ (this follows from Proposition~$3$ in \cite{2014:RikosHadj}, which actually provides a bound for an unobstructed ``leading mass'' to reach every other node). %\footnote{In Proposition~3 in \cite{2014:RikosHadj} the authors show that the number of iterations required for a packet, which is transmitted between nodes in a round-robin fashion over a directed topology, to reach every node in the network is bounded by $m^2$, where $m$ is the number of edges in the network.}). 
Even if the mass gets obstructed at some node, this means that it is no longer the ``leading mass''; in fact, the new ``leading mass'' passed by this node earlier on and has already followed the same path that the former ``leading mass'' would have followed. 
Let us assume now that we execute Algorithm~\ref{algorithm1} for $m^2(L_{\max}+1)$ time steps, where $L_{\max} = \max_{v_j \in \mathcal{V}} L_j$.
During the $m^2(L_{\max}+1)$ time steps each node $v_j$ will receive at least $(L_{\max}+1)$ times a set of nonzero mass variables from its in-neighbors that are equal to the ``leading mass''. 
Since this set of mass variables is equal to the ``leading mass'', from Lemma~\ref{first_lemma}, we have that the event-triggered conditions of each node $v_j$ will hold for at least $(L_{\max}+1)$ events during these $m^2(L_{\max}+1)$ time steps. 
This means that each node $v_j$ adds the offset $u_j[l_j]$ to its mass variables for $(L_{\max}+1)$ events. 
Since we have that $\sum_{l_j=0}^{L_j} u_j[l_j] = -u_j$, from Lemma~\ref{second_lemma}, after $m^2(L_{\max}+1)$ time steps, the accumulated amount of offset each node $v_j$ has injected to the network becomes equal to zero. 
As a result we have that  $\sum_{v_j \in \mathcal{V}} \widetilde{y}_j[m^2(L_{\max}+1)] = \sum_{v_j \in \mathcal{V}} y_j[0]$
and $ \sum_{v_j \in \mathcal{V}} z_j[m^2(L_{\max}+1)] = \sum_{v_j \in \mathcal{V}} z_j[0]$.
After executing Algorithm~\ref{algorithm1} for an additional number of $nm^2$ time steps, we have that the convergence analysis of our protocol becomes identical to the analysis presented in \cite[Proposition~$1$]{2020:Rikos_Mass_Accum} where during the additional $nm^2$ we have that either (a) the ``leading masses'' never merge (because they all move simultaneously) or (b) there are at most $n-1$ ``mergings'' of mass variables (each merging occurring after at most $m^2$ time steps). 
As a result, after $\mathcal{S}_t$ iterations, where {$\mathcal{S}_t \leq m^2(L_{\max} + 1 + n)$}, we are guaranteed that the offset each node has injected into the network becomes equal to zero, and sufficient ``mergings'' (at most $n-1$) occurred, so that nodes will be able to calculate the average of their states. $\hfill$ $\qed$

% ===============================================
%
% Theorem 2
%
% ===============================================
\section{Proof of Proposition~\ref{prop:1}} \label{proof:prop1}

Let us assume that node $v_j$ follows the privacy-preserving protocol. 
We consider the following simple scenarios, which constitute the building blocks of the directed network, due to the token-based nature of the privacy-preserving protocol. 
\begin{list5}
\item[1)] It is easy to observe that if all the in- and out-neighbors of node $v_j$ are curious and they communicate with each other, it is not possible for this node to keep its privacy. 
At initialization, the curious nodes will know $\widetilde{y}_j[0]$. At every step, they will know what node $v_j$ has received and they will be able to extract the offset added. Hence, after $(L_j +1)$ updates from node $v_j$, the curious nodes will be able to compute the initial offset, since the initial offset satisfies \eqref{Offset_value_1e}; hence, privacy of the initial state is not preserved. 
As a result, at least one neighbor that is not curious is needed.
\item[2)] Suppose that there exists at least one in-neighbor, say $v_i$, that is not curious, but it does not follow the privacy-preserving protocol; all other in- and out-neighbors of node $v_j$ and node $v_i$ are assumed (as a worst-case assumption) to be curious. 
If at initialization, node $v_i$ first transmits to node $v_j$, then the curious nodes will not be able to infer its initial state, since they cannot distinguish the initial states of $v_i$ and $v_j$; even if in the end the curious nodes are able to obtain all the offsets, they can only infer the sum of the initial states of node $v_j$ and node $v_i$, but not their individual states, i.e., \emph{both} nodes preserve their privacy. 
If, however, $v_j$ is not contacted by node $v_i$ at initialization, then the curious nodes will be able to extract the initial condition and know all the inputs (and hence outputs) of node $v_i$ (recall that node $v_i$ does not follow the privacy-preserving protocol). 
Thus, curious nodes will be able to infer all the states that node $v_i$ transmits to $v_j$ and, as a consequence, \emph{none} of the nodes will preserve its privacy. This suggests that, if a node $v_i$ trusts that an out-neighbor $v_j$ is not curious and follows the privacy-preserving protocol, then this out-neighbor should be prioritized, i.e.,  $P_{ji}=0$.
\item[3)] Let us consider the case for which there exists one out-neighbor of node $v_j$, say $v_l$, that is neither curious nor following the privacy-preserving protocol, and all other in- and out-neighbors of both nodes are curious. 
Since curious nodes can infer the input of node $v_l$ from its output, then they will be able to extract the messages of node $v_j$ in the same way as if a curious node was directly connected to node $v_j$. Hence, privacy of node $v_j$ cannot be preserved.
\item[4)] If we assume that there exists at least one in-neighbor, say $v_{\ell}$, that follows the privacy-preserving protocol, the curious nodes will not be able to infer the initial state of node $v_j$ due to the offsets node $v_{\ell}$ transmits to $v_j$ and, therefore, both $v_{\ell}$ and $v_j$ retain their privacy.
\end{list5}
From these discussions, we can deduce that it is \emph{sufficient} that conditions a) and b) in Proposition~\ref{prop:1} are satisfied. 
In such cases, the initial states of both nodes are protected (though the sum of these states may be exposed).
This means that the initial state of the node that follows the protocol cannot be inferred exactly, which according to Definition~\ref{Definition_Quant_Privacy} implies that its privacy is preserved. 
Note that it is also sufficient if a node that does not follow the privacy-preserving protocol first selects an out-neighbor that does follow the protocol to transmit its state; see item 5) in the initialization stage of Algorithm~\ref{algorithm1}.  $\hfill$ $\qed$

% ===============================================
%
% Theorem 3
%
% ===============================================
\section{Proof of Theorem~\ref{Privacy_Conver_Quant_Av_Alg2}}\label{proof:theorem3}

During the Initialization steps of Algorithm~\ref{algorithm2}, we have that each node $v_j \in \mathcal{V}_p$ chooses an integer value $u^{(l)}_j \in \mathbb{Z}$ for every out-neighbor $v_l \in \mathcal{N}_j^+$, respectively, while it sets $u_j = - \sum_{v_l \in \mathcal{N}_j^+} u^{(l)}_j$. 
Then, it transmits the chosen integer value $u^{(l)}_j$ to every out-neighbor $v_l \in \mathcal{N}_j^+$. 
Finally, it sets its initial state as 
$ 
\widetilde{y}_j[0] = y_j[0] + u_j + \sum_{v_i \in \mathcal{N}_j^-} u^{(j)}_i, 
$
and proceeds with executing the protocol described in Section~\ref{Prel_Aver}. 
Focusing on $\widetilde{y}_j[0]$ we have that 
\begin{equation}\label{init_decompose1}
\sum_{v_j \in \mathcal{V}} \widetilde{y}_j[0] = \sum_{v_j \in \mathcal{V}} \bigl ( y_j[0] + u_j + \sum_{v_i \in \mathcal{N}_j^-} u^{(j)}_i \bigr ).  
\end{equation}
From \eqref{init_decompose1}, we have  
\begin{eqnarray}\label{init_decompose2}
\sum_{v_j \in \mathcal{V}} \bigl ( y_j[0] + u_j + \sum_{v_i \in \mathcal{N}_j^-} u^{(j)}_i \bigr ) & = \nonumber \\
\sum_{v_j \in \mathcal{V}} y_j[0] + \sum_{v_j \in \mathcal{V}} u_j + \sum_{v_j \in \mathcal{V}} \bigl ( \sum_{v_i \in \mathcal{N}_j^-} u^{(j)}_i \bigr ) &  .  \;
\end{eqnarray}
Analyzing the second part of \eqref{init_decompose2} we have that from \eqref{Offset_value_2a}, \eqref{Offset_value_2b} it holds: 
\begin{equation}\label{init_decompose3}
\sum_{v_j \in \mathcal{V}} u_j = - \sum_{v_j \in \mathcal{V}} \bigl ( \sum_{v_i \in \mathcal{N}_j^-} u^{(j)}_i \bigr ) .   
\end{equation}
From \eqref{init_decompose3}, we have that \eqref{init_decompose2} becomes 
\begin{equation}\label{init_decompose5}
\sum_{v_j \in \mathcal{V}} \bigl ( y_j[0] + u_j + \sum_{v_i \in \mathcal{N}_j^-} u^{(j)}_i \bigr ) = \sum_{v_j \in \mathcal{V}} y_j[0] .  
\end{equation}
As a result, from \eqref{init_decompose1} and \eqref{init_decompose5} we have, 
$
\sum_{v_j \in \mathcal{V}} \widetilde{y}_j[0] = \sum_{v_j \in \mathcal{V}} y_j[0] ,  
$
which means that the Initialization Steps of Algorithm~\ref{algorithm2} not only preserve the privacy of each node's initial state, i.e., each node's $v_j$ initial state becomes $\widetilde{y}_j[0]$ instead of $y_j[0]$, but also preserve the sum of the initial states. 
This means that, during the Initialization Steps of Algorithm~\ref{algorithm2} each node obtains the state $\widetilde{y}_j[0]$ and during the Iteration Steps, the protocol described in Section~\ref{Prel_Aver} is executed. 
As a result, from Theorem~\ref{Conver_Quant_Av}, Algorithm~\ref{algorithm2} allows each node $v_j \in \mathcal{V}$ to reach quantized average consensus (i.e., $v_j$'s state variables fulfil (\ref{alpha_z_y}) and (\ref{alpha_q})) after a finite number of steps $\mathcal{S}_t$, bounded by $\mathcal{S}_t \leq nm^2$, where $n$ is the number of nodes and $m$ is the number of edges in the network.   $\hfill$ $\qed$

\end{document}